\documentclass[12pt,reqno]{amsart}

\usepackage{amsmath}
\usepackage{amsthm}
\usepackage{amsfonts}
\usepackage{amssymb}
\usepackage{latexsym}
\usepackage{amscd}
\usepackage{stmaryrd}
\usepackage{amsbsy}

\usepackage{txfonts}

\allowdisplaybreaks
\newcommand{\nn}{\nonumber}
\newcommand{\C}{{\mathbb C}}

\newcommand{\slt}{\mathfrak{sl}_2}

\newcommand{\Ad}{{\rm Ad}}

\newcommand{\ka}{\kappa}
\newcommand{\al}{\alpha}
\newcommand{\be}{\beta}
\newcommand{\ep}{\epsilon}

\numberwithin{equation}{section}

\theoremstyle{plain}
\newtheorem{thm}{Theorem}[section]

\newtheorem{prop}[thm]{Proposition}

\newtheorem{dfn}[thm]{Definition}

\newtheorem{re}[thm]{Remark}

\begin{document}

\title{Symmetries of quantum Lax equations for the Painlev\'e equations
}
\date{\today}

\author{Hajime Nagoya} 
\address{Department of Mathematics, Kobe University,
 Kobe 657-8501, Japan, 
 Research Fellow of the Japan Society for the Promotion of Science}
\email{nagoya@math.kobe-u.ac.jp}

\author{Yasuhiko Yamada}
\address{Department of Mathematics, Kobe University,
 Kobe 657-8501, Japan}
\email{yamaday@math.kobe-u.ac.jp}

\begin{abstract}
The Painlev\'e equations can be written as Hamiltonian systems with affine Weyl group symmetries.  
A canonical quantization of the Painlev\'e equations preserving the affine Weyl group symmetries has been studied. 
While, the Painlev\'e equations are isomonodromic equations for certain second-order  linear differential equations. 
In this paper, 
we introduce a canonical quantization of Lax equations for the Painlev\'e equations and construct symmetries 
of the quantum Lax equations. We also show that our quantum Lax equations are derived from 
Virasoro conformal field theory. 
\end{abstract}
\maketitle



%%%%%%%%%%%%%%%%%%%%%%%%%%%%%%%%%%%%%%%%%%%%%%
\section{Introduction}
It is known that the Painlev\'e equations are Hamiltonian systems and, except for the first one, admit the affine Weyl 
group actions, as B\"acklund transformations \cite{O}. For example, 
the second Painlev\'e equation $\mathrm{P_{II}}(\alpha)$ $(\alpha \in \C)$ is the 
Hamiltonian system:
\begin{equation*}
{dq\over dt}={\partial H_{\mathrm{II}}\over \partial p},
\quad {dp\over dt}=-{\partial H_{\mathrm{II}}\over \partial q}, 
\end{equation*}
where 
\begin{equation*}
H_{\mathrm{II}}(q,p,t,\al)={p^2\over 2}-\left(q^2+{t\over 2}\right)p-\al q.  
\end{equation*}
Let $(q,p)$ be a solution to $\mathrm{P_{II}}(\al)$. Then, 
 birational canonical transformations defined by 
\begin{align*}
&s(q,p)=(q+{\al\over p}, p),
\\
&\pi(q,p)=(-q,-p+2q^2+t), 
\end{align*}
give solutions to $\mathrm{P_{II}}(-\al)$, $\mathrm{P_{II}}(1-\al)$, respectively. 
 The B\"acklund transformation group generated by  $s$, $\pi$  is 
equivalent to the extended affine Weyl group of type $A_1^{(1)}$. 
%%%%%%%%%%%%%%%%%%%%%%%%%%%%%%%%%%%%%%%%%%%%%%

Since the Painlev\'e equations are Hamiltonian systems, their quantization can be considered naturally. 
A canonical quantization of the Painlev\'e equations preserving the affine Weyl group actions have 
been studied \cite{N QNY}, \cite{N HGS}, \cite{N Weyl} (see also \cite{JNS}, \cite{NGR}, \cite{N QPVI}). For example, the quantum 
second Painlev\'e equation $\mathrm{QP_{II}}$ can be written as the time-dependent Schr\"odinger  equation:
\begin{align*}
\ka \frac{\partial}{\partial t}\Psi(t,x)=&H_{\mathrm{II}}\left(
x,{\partial \over \partial x},t,\al
\right)
\Psi(t,x)
\\
=&\left(  {1\over 2}\left(   {\partial\over \partial x}\right)^2 -x{\partial\over \partial x}x-{t\over 2} 
{\partial\over \partial x}-\al x
\right)\Psi(t,x). 
\end{align*}
B\"acklund transformations of $\mathrm{QP_{II}}$ are realized by  the Euler transformation (or the Riemann-Liouville integral) and a gauge transformation. 
Let $\Psi(t,x)$ be a solution to $\mathrm{QP_{II}}(\al)$. Then, transformations of a solution $\Psi(t,x)$ defined by 
\begin{align*}
&s\left(\Psi(t,x)\right)=\int_\Delta (x-u)^{\al-1}\Psi(t,u)du,
\\
&\pi\left(  \Psi(t,x) \right)=\exp\left(  {2\over 3}x^3+xt \right)\Psi(t,-x), 
\end{align*}
with an appropriate cycle $\Delta$, 
are solutions to $\mathrm{QP_{II}}(-\al)$, $\mathrm{QP_{II}}(-\ka-\al)$, respectively. 
Similarly, the affine Weyl group symmetries for the quantum Painlev\'e equations $\mathrm{QP_{III}}$-
$\mathrm{QP_{VI}}$ were realized by using  gauge transformations and the Laplace transformation \cite{N Weyl}.
In both the classical and quantum cases, the affine Weyl group symmetries play an important  role to study 
special solutions to the systems. 

%Symmetries of the Lax equations for the Painlev\'e equations have been also studied; 
%Schlesinger transformations \cite{JM}, integral transformations \cite{Takemura}, \cite{Kawakami}. 

On the other hand, the Painlev\'e equations describe the isomonodromic deformation for certain second-order
linear differential equations \cite{JM}. 
Since this fact is crucial for the Painlev\'e equations, 
it will be important to study its quantization.
In the present paper, we introduce  quantum Lax equations\footnote{
We call the linear differential equations (the Lax auxiliary linear problems)
simply as Lax equations.}
and study
their symmetries. In doing this, a useful fact is that the classical Lax equation can be written
concisely in terms of the quantum and classical Hamiltonians. 
For example, the Lax equation for the second Painlev\'e equation $\mathrm{P_{II}}(\al+1)$
can be written as
\begin{equation*}
\left(H_{\mathrm{II}}\left(
x, {\partial \over \partial x},t,\al
\right)
-H_{\mathrm{II}}(q,p,t,\al+1)-{1\over 2(x-q)}\left(
{\partial \over \partial x}-p
\right)\right)y(x)=0,
\end{equation*}
and a natural quantization of this gives the following quantum Lax equation:
\begin{equation*}
\left(H_{\mathrm{II}}\left(
x, \ep_1{\partial \over \partial x},t,\al
\right)
-H_{\mathrm{II}}\left(
q, \ep_2{\partial \over \partial q},t,\al+\ep_1-\ep_2
\right)
-{\ep_1-\ep_2\over 2(x-q)}\left(
\ep_1{\partial \over \partial x}-\ep_2{\partial \over \partial q}
\right)
\right)
\Phi(x,q)=0. 
\end{equation*}
The symmetry of these quantum Lax equations can be derived by using the symmetry properties of the quantum Hamiltonians studied in \cite{N Weyl}.
Taking the classical limit of the quantum Lax equations as  
$\epsilon_2\rightarrow 0$  with
$\epsilon_2 \partial/\partial q \rightarrow p$, we recover the classical Lax equations and symmetries of them. On realization 
of symmetries of the classical Lax equations, see \cite{Takemura}, \cite{Kawakami} and references therein, for example. 
We also derive the quantum Lax equations from  Virasoro conformal field theory with two null fields
at $x$ and $q$. Note that the quantum Painlev\'e equations are derived from the conformal field theory with one null field 
\cite{BPZ}, 
\cite{N HGS}, \cite{AFKMY}. 

%%%%%%%%%%%%%%%%%%%%%%%%%%%%%%%%%%%%%%%%%%%%%%

Similarly in the case of the quantum 
Painlev\'e equations \cite{N Weyl}, 
symmetries constructed in this paper generate solutions to the quantum Lax equations. 
We shall  investigate solutions to the quantum Lax equations in the forthcoming paper. 

%%%%%%%%%%%%%%%%%%%%%%%%%%%%%%%%%%%%%%%%%%%%%

%%%%%%%%%%%%%%%%%%%%%%%%%%%%%%%%%%%%%%%%%%%%%

The remainder of this paper is organized as follows. In section 2, we introduce quantum 
Lax equations for the Painlev\'e equations. After recalling symmetries of the quantum Painlev\'e equations, 
we define transformations and show that those are B\"acklund transformations for the quantum Lax equations. 
In section 3, we derive quantum Lax equations introduced in section 2 from  Virasoro conformal field theory. 
In appendix, we summarize the known results for the classical case.

%%%%%%%%%%%%%%%%%%%%%%%%%%%%%%%%%%%%%%%%%%%%%%
\begin{re}
It is known that  %the quantum sixth Painlev\'e equation is a Belavin-Polykov-Zamolodchikov(BPZ) equation 
%for five points. 
%It is known that 
the quantum Painlev\'e equations with $\ka=1$ have a relation to 
  corresponding classical Lax equations  
\cite{Suleimanov}, \cite{Novikov}, \cite{ZZ}. More precisely, the wave functions of the classical Lax equations 
multiplied by the tau functions of the Painlev\'e equations are solutions to the quantum Painlev\'e equations with 
$\kappa=1$. 
This means that the classical Lax equations are related to the conformal field theory with the central charge $c=1$.
In \cite{GIL}, the tau function of the classical sixth Painlev\'e equation is interpreted as a four points correlation function in the
conformal field theory with $c=1$.
\end{re}

%%%%%%%%%%%%%%%%%%%%%%%%%%%%%%%%%%%%%%%%%%%%%%
\section{Symmetry}

In this section, we introduce the quantum Lax equations for the Painlev\'e 
equations and describe symmetries of them.  In order to construct B\"acklund transformations of the quantum Lax equations, 
we use B\"acklund transformations of the quantum Painlev\'e equations.

%%%%%%%%%%%%%%%%%%%%%%%%%%%%%%%%%%%%%%
\subsection{$\mathrm{P}_\mathrm{{VI}}$ case}

Let $\mathcal{K}$ be the skew field over $\mathbb{C}$ defined by 
the generators $x$, $y$, $q$, $p$, $t$, $d$, $\alpha_i$ ($0\le i\le 4$), $\ep_1,\ep_2$, 
and 
the commutation relations 
\begin{align*}
&[y,x]=\epsilon_1,\quad [p,q]=\epsilon_2,\quad [d,t]=1, 
\end{align*}
%where $\epsilon_1$, $\epsilon_2$ are complex parameters 
and 
the other commutation relations are zero, and a relation $\al_0+\al_1+2\al_2+\al_3+\al_4=-\ep_1+\ep_2$. 

Let $H^x_\mathrm{{VI}}(\al)$ ($\al=(\al_0,\al_1,\al_2,\al_3,\al_4)$)
be the  Hamiltonian for  the quantum sixth  Painlev\'e  equation defined by 
\begin{align*}
H^x_\mathrm{{VI}}(\al)=&x(x-1)(x-t)\left(y-{\al_4-\ep_1\over x}-{\al_3-\ep_1\over x-1}
-{\al_0-\ep_2\over x-t}\right)y
\\
&+(\al_2+\ep_1)(\al_1+\al_2+\ep_1)x.\quad 
\end{align*}
  Let $H^q_\mathrm{{VI}}(\al)$ be defined by replacing 
$x$, $y$, $\ep_1$, $\ep_2$ in $H^x_\mathrm{{VI}}(\al)$ with $q$, $p$, $\ep_2$, $\ep_1$, respectively. 

Let us introduce the quantum Lax operators $L_\mathrm{{VI}}(\al)$ and $B_\mathrm{{VI}}(\al)$ for the sixth Painlev\'e equation defined by 
\begin{align}
L_\mathrm{{VI}}(\al)=&H^x_\mathrm{{VI}}(\al_0,\al_1,\al_2,\al_3,\al_4)-H^q_\mathrm{{VI}}(\al_0,\al_1,\al_2+\ka,\al_3,\al_4)
\nn
\\
&-{\ka \over x-q}\left(x(x-1)(q-t)y-q(q-1)(x-t)p\right), \label{eq Lax LVI}
\\
B_\mathrm{{VI}}(\al)=&\ep_2 H^x_\mathrm{{VI}}(\al_0,\al_1,\al_2,\al_3,\al_4)-\ep_1 H^q_\mathrm{{VI}}(\al_0,\al_1,\al_2+\ka,\al_3,\al_4)
-\ka\ep_1\ep_2 t(t-1) d.  \label{eq Lax BVI}
\end{align}
Here $\ka=\ep_1-\ep_2$.   We use this notation throughout the paper. 

Let us recall 
  the extended affine Weyl group 
$\widetilde{W}(D_4^{(1)})$ symmetry of the quantum sixth Painlev\'e equation. 
Here, $\widetilde{W}(D_4^{(1)})=W(D_4^{(1)})\rtimes G$, where 
$W(D_4^{(1)})=\langle s_0,s_1,s_2,s_3,s_4\rangle$ is the affine Weyl group of type $D_4^{(1)}$ 
and $G=\langle \sigma_1,\sigma_2,\sigma_3 \rangle$ 
is the  automorphism group of the  Dynkin diagram of type $D_4^{(1)}$.   

\begin{dfn}[cf. \cite{N QPVI}]\label{dfn QPVI D}
Let the automorphisms $s^q$ for $s\in\{ s_0,s_1,s_2,s_3,s_4,\sigma_1,\sigma_2,\sigma_3 \}$ 
on $\mathcal{K}$ be defined by the following table:  
\begin{center}
\begin{tabular}{|c|ccccc|cc|cc|}
\hline
$z$&$\alpha_0$&$\alpha_1$&$\alpha_2$&$ \alpha_3$&$\alpha_4$&$q$&$p$&$t$&$d$
\\ \hline
$s_0^q(z)$ & $-\alpha_0$ & $\alpha_1$ & $\alpha_2+\alpha_0$ & $\alpha_3$
 &$\alpha_4$ & $q$ & $p-{\alpha_0\over q-t}$&$t$&$d+\frac{\al_0/\ep_2}{q-t}$
\\
$s_1^q(z)$&$\alpha_0$ & $-\alpha_1$ & $\alpha_2+\alpha_1$ & $\alpha_3$
 &$\alpha_4$ & $q$ & $p$&$t$&$d$
\\
$s_2^q(z)$&$\alpha_0+\alpha_2$ & $\alpha_1+\alpha_2$ & $-\alpha_2$ 
& $\alpha_3+\alpha_2$
 &$\alpha_4+\alpha_2$ & $q+{\alpha_2\over p}$ & $p$&$t$&$d$
\\
$s_3^q(z)$&$\alpha_0$ & $\alpha_1$ & $\alpha_2+\alpha_3$ & $-\alpha_3$
 &$\alpha_4$ & $q$ & $p-{\alpha_3\over q-1}$&$t$&$d$
\\
$s_4^q(z)$&$\alpha_0$ & $\alpha_1$ & $\alpha_2+\alpha_4$ & $\alpha_3$
 &$-\alpha_4$ & $q$ & $p-{\alpha_4\over q}$&$t$&$d$
\\ \hline
%$\pi_1(x)$ &$\alpha_3$ & $\alpha_4$& $\alpha_2$ &$\alpha_0$ &$\alpha_1$&${t\over q}$
%&$-{q(pq+\alpha_2)\over t}$ &$t$
%\\
%$\pi_2(x)$&$\alpha_1$ &$\alpha_0$ &$\alpha_2$&$\alpha_4$&$\alpha_3$&${(q-1)t\over q-t}$
%&$-{(q-t)(p(q-t)+\alpha_2)\over t(t-1)}$& $t$
%\\ \hline
$\sigma_1^q(z)$&$\alpha_0$&$\alpha_1$&$\alpha_2$&$\alpha_4$&$\alpha_3$&$1-q$&$-p$&
$1-t$&$-d$
\\
$\sigma_2^q(z)$&$\alpha_0$&$\alpha_4$&$\alpha_2$&$\alpha_3$&$\alpha_1$&${1\over q}$&
$-q(pq+\alpha_2)$&${1\over t}$&$-t^2d$
\\
$\sigma_3^q(z)$&$\alpha_4$&$\alpha_1$&$\alpha_2$&$\alpha_3$&$\alpha_0$&${t-q\over t-1}$&
$-(t-1)p$&${t\over t-1}$&$\begin{matrix}(1-t)(q-1)p\\-(t-1)^2d\end{matrix}$
\\ \hline
\end{tabular}
\end{center}
% The automorphisms $s_i^x$ ($i=0,1,2,3,5$) and 
%$\sigma_i^x$ ($i=1,2,3$) are defined by replacing $q$, $p$, $\ep_2$ in $s_i^q$, $\sigma_i^q$ with $x$, $y$, $\ep_1$. 
%More precisely, 
%Moreover, the automorphisms $s_i$ ($i=0,1,2,3,4$) and $\sigma_i$ ($i=1,2,3$) 
%satisfy the fundamental relations of $\widetilde{W}(D_4^{(1)})$. 
\end{dfn}

Let 
$s_i$ ($0\le i\le 4$) 
be 
the automorphisms on $\mathcal{K}$  defined by  
$s_i(\alpha_j)=s^q_i(\alpha_j)$ 
for $ j=0,\ldots, 4$, and $s_i(f)=f$, 
for $f=x,y, q ,p, t, d$, and let $\sigma_i$ ($1\le i\le 3$) be the 
automorphisms on $\mathcal{K}$ defined by $\sigma_i(\alpha_j)=\sigma^q_i(\alpha_j)$ for $ j=0,\ldots, 4$, 
and $\sigma_i(f)=f$ 
for $f=x,y, q ,p, t, d$. 

The automorphisms $s^q$ for $s\in\{ s_0,s_1,s_2,s_3,s_4,\sigma_1,\sigma_2,\sigma_3 \}$ are expressed 
as compositions of transformations $s$ for parameters  and transformations $R^q_s$ 
for variables, that is, $s^q=s\circ R^q_s$ (Theorem 2.4 \cite{N Weyl}). 
The automorphism $R^q_s$ is a B\"acklund transformation for the quantum sixth 
Painlev\'e equation, which transforms a solution  with the parameter $\al$ to a solution with the parameter $s(\al)$. 
As for birational actions of the 
Weyl group of any symmetrizable generalized Cartan matrix, see \cite{K} and reference therein.

Let $\mathcal{L}_x$, $\mathcal{L}_q$  be the Laplace transformations  on $\mathcal{K}$ 
with respect to $x$, $q$, respectively, defined by 
\begin{equation*}
\mathcal{L}_x\left(y\right)=x,\quad
\mathcal{L}_x\left(x\right)=-y,\quad \mathcal{L}_q\left(p\right)=q,\quad
\mathcal{L}_q\left(q\right)=-p. 
\end{equation*}
Let $\Ad( (x-c)^{\beta/\ep_1})$ for ( $c\in \C$, $\beta\in\C$) be the 
gauge transformations on $\mathcal{K}$ defined by 
\begin{equation*}
\Ad( (x-c)^{\beta})\left(y\right)=y
-\frac{\beta}{x-c}.  
\end{equation*}
Let $\Ad( (x-t)^{\beta/\ep_1})$ for ($\beta\in\C$) be the 
gauge transformations on $\mathcal{K}$ defined by 
\begin{equation*}
\Ad( (x-t)^{\beta/\ep_1})\left(y\right)=y
-\frac{\beta}{x-t},\quad 
\Ad( (x-t)^{\beta/\ep_1})\left(d\right)=d
+\frac{\beta/\ep_1}{x-t}.  
\end{equation*}
Here, we have omitted to write the transformation on the variables if it acts identically.
The automorphisms $\Ad\left((q-c)^{\be/\ep_2}\right)$, $\Ad\left((q-t)^{\be/\ep_2}\right)$ are defined in the 
same way above.

\begin{dfn}[cf. \cite{N Weyl}]
Let the automorphisms $R^x_{s_i}(\al_i)$ ($i=0,1,2,3,4$) and $R_{\sigma_i}$ ($i=1,3$), $R^x_{\sigma_2}(\al_2)$ 
on $\mathcal{K}$
be defined by 
\begin{align*}
&R^x_{s_0}(\al_0)=  \Ad\left( (x-t)^{-{\alpha_0\over \ep_1}}\right), 
\quad
R^x_{s_1}(\al_1)=\mathrm{id},\quad 
R^x_{s_2}(\al_2)=\mathcal{L}_x^{-1}\circ \Ad\left( x^{-{\alpha_2\over \ep_1}}\right) \circ \mathcal{L}_x,  
\\
&R^x_{s_3}(\al_3)= \Ad\left( (x-1)^{-{\alpha_3\over \ep_1}}\right), 
\quad
R^x_{s_4}(\al_4)= \Ad\left( x^{-{\alpha_4\over \ep_1}}\right), 
\\
&R_{\sigma_1}=\left(x\mapsto 1-x, q\mapsto 1-q,t\mapsto 1-t\right),
\quad R^x_{\sigma_2}(\al_2)=R^x_{s_4}(\al_2+\ep_1)\circ\left(
x\mapsto \frac{1}{x}, q\mapsto \frac{1}{q}, t\mapsto\frac{1}{t}\right),\\
&R_{\sigma_3}= \left( x\mapsto \frac{t-x}{t-1},  q\mapsto \frac{t-q}{t-1}, t\mapsto \frac{t}{t-1}\right).
\end{align*}
Here, $\left(x\mapsto f(x,t), t\mapsto g(x,t)\right)$ stands for a transformation of variables. 
The automorphisms $R^q_{s_i}(\al_i)$ ($i=0,1,2,3,4$), $R^q_{\sigma_2}(\al_2)$  are defined by replacing $x$,  $\ep_1$ in $R_{s_i}^x(\al_i)$, 
$R^x_{\sigma_2}(\al_2)$    
with $q$,  $\ep_2$, respectively. 
\end{dfn}

\begin{prop}
[\cite{N QPVI}]\label{prop H_VI sym}
 The automorphisms $R^x_{s_i}(\al_i)$ ($i=0,1,2,3,4$), $R_{\sigma_i}$ ($i=1,3$), $R^x_{\sigma_2}(\al_2)$ preserve  the Hamiltonian 
$H^x_\mathrm{{VI}}(\al)$ in the following sense:
\begin{align*}
&R^x_{s_i}(\al_i)(H^x_\mathrm{{VI}}(\al))=H^x_\mathrm{{VI}}(s_i(\al))+C_{s_i}, 
\quad
R_{\sigma_1}(H^x_\mathrm{{VI}}(\al))=-H^x_\mathrm{{VI}}(\sigma_1(\al))+C_{\sigma_1},\quad 
\\
&R^x_{\sigma_2}(\al_2)(H^x_\mathrm{{VI}}(\al))={1\over t}H^x_\mathrm{{VI}}(\sigma_2(\al))+C_{\sigma_2}, \quad 
R_{\sigma_3}(H^x_\mathrm{{VI}}(\al))={1\over 1-t}H^x_\mathrm{{VI}}(\sigma_3(\al))+C_{\sigma_3}, 
\end{align*}
where 
\begin{align*}
C_{s_0}=&
\al_0\left(\al_4-\ep_1+\ka x+\ka {x(x-1)\over t-x}\right),\quad 
\\
C_{s_1}=&0,
\\
C_{s_2}=&
\al_2(\al_3+\al_1+\al_2+\ep_1+(\al_0+\al_1+\al_2+\ep_1+\ka)t),
\\
C_{s_3}=&
\al_3((\al_4-\ep_1)t-\ka x),\quad 
\\
C_{s_4}=&
\al_4(\al_0-\ep_2+(\al_3 - \ep_1)t ),\quad
\\
C_{\sigma_1}=&(\al_2+\ep_1)(\al_1+\al_2+\ep_1), 
\\
C_{\sigma_2}=&{1\over t}(\al_2+\ep_1)(\al_0+\al_1+\al_2+\ka+t(\al_1+\al_2+\al_3),
\\
C_{\sigma_3}=&{t\over t-1}\left(  (\al_2+\ep_1)(\al_1+\al_2+\ep_1)
-\ka(x-1)y \right). 
\end{align*}
\end{prop}
By definition, the automorphisms $R_{s_i}^q(\al_i)$ ($i=0,1,2,3,4$) and $R_{\sigma_i}$ ($i=1,3$), $R^q_{\sigma_2}(\al_2)$ 
 act the Hamiltonian $H^q_\mathrm{{VI}}(\al)$ in the same way above. 

Let $D(\al_2)$ be defined by 
\begin{equation}\label{eq D}
D(\al_2)=
y p+{\al_2+\ep_1\over x-q}y+{\al_2+\ep_1\over q-x}p. 
\end{equation}
 We use this notation throughout the paper. 

\begin{dfn}
Let the automorphisms $R_{s_i}$ ($i=0,1,3,4$), $R_{\sigma_2}$, $T_{s_0s_1s_3s_4s_2}$ and $S$ on $\mathcal{K}$ be defined by 
\begin{align*}
&R_{s_i}=R^x_{s_i}(\al_i)R^q_{s_i}(\al_i),\quad R_{\sigma_2}=R^q_{s_4}(\al_2+\ep_1)R^x_{\sigma_2},
\\
&T_{s_0s_1s_3s_4s_2}=R^x_{s_2}(-\al_2-\ka)R^x_{s_0}(\al_0)R^x_{s_1}(\al_1)R^x_{s_3}(\al_3)R^x_{s_4}(\al_4)
R^q_{s_0}(\bar{\al_0})R^q_{s_1}(\bar{\al_1})R^q_{s_3}(\bar{\al_3})R^q_{s_4}(\bar{\al_4})R^q_{s_2}(\al_2+\ka),
\\
&S=\Ad(D(\al_2)^{-1})R^x_{s_2}(\al_2)R^q_{s_2}(\al_2+\ka), 
\end{align*}
where $\bar{\al_i}=-s_0s_1s_3s_4s_2(\al_0)=\al_i+\al_2+\ka$ for $i=0,1,3,4$. 
\end{dfn}
These automorphisms $R_{s_i}$ ($i=0,1,3,4$), $R_{\sigma_2}$  and $T_{s_0s_1s_3s_4s_2}$ are naturally given by 
looking at the change of parameters when the automorphisms $R^x_{s_i}(\al_i)$, $R^q_{s_i}(\al_i)$ act the quantum Lax operators. 
%While, we do not have simple explanation for the operator $S$. 

\begin{thm}\label{thm VI}

The automorphisms $R_{s_i}$ ($i=0,1,3,4$), $T_{s_0s_1s_3s_4s_2}$ and $S$ act the quantum Lax operators 
$L_\mathrm{{VI}}(\al)$ and $B_\mathrm{{VI}}(\al)$ as follows. 

For $s\in\{ s_0,s_1,s_3,s_4,\sigma_1,\sigma_2,\sigma_3 \}$,  
\begin{equation*}
R_{s}\left(
L_\mathrm{{VI}}(\al),B_\mathrm{{VI}}(\al)
\right)=c_s\left(
L_\mathrm{{VI}}(s(\al)),B_\mathrm{{VI}}(s(\al))+f_{s}
\right), 
\end{equation*}
where
\begin{equation*}
c_{s_i}=1\quad (i=0,1,3,4),\quad c_{\sigma_1}=-1,\quad c_{\sigma_2}={1\over t},\quad c_{\sigma_3}={1\over 1-t},
\end{equation*}
and 
\begin{align*}
&f_{s_0}=-\ka \al_0(\al_4+(t-1)(\ep_1+\ep_2)),
\\
&f_{s_1}=0,
\\
&f_{s_3}=-\ka \al_3\al_4t,
\\
&f_{s_4}=-\ka\al_4(\al_0-\ep_1-\ep_2+\al_3 t),
\\
&f_{\sigma_1}=\ka(\al_2+\ep_1)(\al_1+\al_2+\ep_1),
\\
&f_{\sigma_2}=-\ka(\al_2+\ep_1)(\al_0+\al_1+\al_2-\ep_2+(\al_1+\al_2+\al_3+\ep_1)t),
\\
&f_{\sigma_3}=\ka(\al_2+\ep_1)(\al_1+\al_2+\ep_1)t. 
\end{align*}

For the automorphism $T_{s_0s_1s_3s_4s_2}$,  
\begin{align}
&l_{T_{s_0s_1s_3s_4s_2}} T_{s_0s_1s_3s_4s_2}\left((x-q)
L_\mathrm{{VI}}(\al)
\right)=(x-q)L_\mathrm{{VI}}\left(s_0s_1s_3s_4s_2
(\al)
\right)
,\label{eq VI euler}
\\
&T_{s_0s_1s_3s_4s_2}\left(
B_\mathrm{{VI}}(\al)
\right)=B_\mathrm{{VI}}\left(s_0s_1s_3s_4s_2
(\al)
\right)+f_{T_{s_0s_1s_3s_4s_2}},\nn
\end{align}
 where 
\begin{align*}
%&l_{T_{s_0s_1s_3s_4s_2}}=R^q_{s_0}(\bar{\al_0})R^q_{s_1}(\bar{\al_1})R^q_{s_3}(\bar{\al_3})R^q_{s_4}(\bar{\al_4})R^q_{s_2}(\al_2+\ka)
%\left(\left(\left(
%(x-q)y+\al_2+\ka
%\right)
%p\right)^{-1}
%\left(
%(x-q)p-\al_2-\ka
%\right)
%y\right),
%\\
&f_{T_{s_0s_1s_3s_4s_2}}=-\ka\left(
(\al_2+\ep_1)(\al_1+\al_2+\al_3+\al_0t)+(\al_2+\ka)(\al_1+\al_2+\ep_1)t
\right),
\end{align*}
and $l_{T_{s_0s_1s_3s_4s_2}}$ is some element in $\mathcal{K}$ whose explicit form is given in the proof.

For the automorphism $S$, 
\begin{align*}
&y p\left(R^x_{s_2}(\al_2)R^q_{s_2}(\al_2+\ka)\left((x-q)L_\mathrm{{VI}}(\al)\right)\right)D(\al_2)
\\
&=\left(
(x-q)y p+(\al_2+\ka-\ep_2)y+(\ep_1-\al_2)p
\right)
\left(
D(\al_2)-{(\al_2+\ep_1)(\ep_1+\ep_2)\over (x-q)^2}
\right)L_\mathrm{{VI}}(\tilde{\al_0},\tilde{\al_1},-\al_2-2\ep_1, \tilde{\al_3},\tilde{\al_4}), 
\\
&S\left(
B_\mathrm{{VI}}(\al)
+f_S
\right)=B_\mathrm{{VI}}(\tilde{\al_0},\tilde{\al_1},-\al_2-2\ep_1, \tilde{\al_3},\tilde{\al_4})
-D(\al_2)^{-1}{2(\al_2+\ep_1)\ep_1\ep_2\over (x-q)^2}L_\mathrm{{VI}}(\tilde{\al_0},\tilde{\al_1},-\al_2-2\ep_1, \tilde{\al_3},\tilde{\al_4}), 
\end{align*}
where $\tilde{\al_i}=\al_i+\al_2+\ep_1$ ($i=0,1,3,4$) and 
\begin{equation*}
f_S=\ka(\al_2+\ep_1)(\al_1+\al_2+\al_3+\ep_1+(\al_0+\al_1+\al_2-\ep_2)t). 
\end{equation*}
\end{thm}

\begin{proof}
A proof follows from direct computation. As an example, we compute \eqref{eq VI euler}
whose precise form is 
\begin{align}
&\left(
(x-q)p-\al_2-\ka
\right)
y
R_{s_2}^x(-\al_2-\ka)R_{s_0}^x(\al_0)R^x_{s_3}(\al_3)R^x_{s_4}(\al_4)\left((x-q)
L_\mathrm{{VI}}(\al)
\right)\nn
\\
&-\left(
(x-q)y+\al_2+\ka
\right)
p
R^2_{s_2}(-\al_2-\ka)R^q_{s_0}(-\bar{\al_0})R^q_{s_3}(-\bar{\al_3})R^q_{s_4}(-\bar{\al_4})\left((x-q)
L_\mathrm{{VI}}(s_0s_1s_3s_4s_2
(\al)
\right)=0.\label{eq VI euler2} 
\end{align}

From Proposition \ref{prop H_VI sym}, we have
\begin{align*}
&R^x_{s_0}(\al_3)R^x_{s_3}(\al_3)R^x_{s_4}(\al_0)\left(
H^x_\mathrm{{VI}}(\al)-{\ka\over x-q}
x(x-1)(q-t)y
\right)
\\
&=H^x_\mathrm{{VI}}(s_0s_1s_3s_4(\al))
-\ep_1((\al_4+\al_3)t+\al_4+\al_0)
\\
&-{\ka\over x-q}\left(
x(x-1)(q-t)y+\al_4 t+\left(
\al_4(q-t-1)+\al_3(q-t)+\al_0(q-1)
\right)x
\right).
\end{align*}
From Proposition  \ref{prop H_VI sym} and above,  we have
\begin{align}
&R^x_{s_2}(-\al_2-\ka)R^x_{s_0}(\al_0)R^x_{s_3}(\al_3)R^x_{s_4}(\al_4)\left((x-q)
H^x_\mathrm{{VI}}(\al)-\ka
x(x-1)(q-t)y
\right)\nn
\\
&=(x-q)\left(
H^x_\mathrm{{VI}}(s_0s_1s_3s_4s_2
(\al))+A_1
\right)\nn
\\
&-\ka((q-t)(x(x-1)y+(\al_2+\ka)(2x-1))+\al_4t+B_1x)\nn
\\
&+{\al_2+\ka\over y}\left(
H^x_\mathrm{{VI}}(s_0s_1s_3s_4s_2
(\al))+A_1-\ka(\al_2-\ep_2)(q-t)-\ka B_1
\right),\label{eq VI euler x}
\end{align}
where 
\begin{align*}
&A_1=-\ep_1((\al_4+\al_3)t+\al_4+\al_0)+(\al_2+\ka)(\al_3+\al_1+\al_2+\ep_2+(\al_0+\al_1+\al_2+\ep_1)t),
\\
&B_1=\al_4(q-t-1)+\al_3(q-t)+\al_0(q-1).
\end{align*}

In a similar way, we have
\begin{align}
&R^q_{s_2}(-\al_2-\ka)R^q_{s_0}(-\bar{\al_0})R^q_{s_3}(-\bar{\al_3})R^q_{s_4}(-\bar{\al_4})\left((x-q)
H^q_\mathrm{{VI}}(s_0s_1s_3s_4s_2
(\al))-\ka
q(q-1)(x-t)p
\right)\nn
\\
&=(x-q)\left(
H^q_\mathrm{{VI}}(\al_0,\al_1,\al_2+\ka,\al_3,\al_4)+A_2
\right)\nn
\\
&-\ka((x-t)(q(q-1)p+(\al_2+\ka)(2q-1))-(\al_4+\al_2+\ka)t-B_2q)\nn
\\
&-{\al_2+\ka\over p}\left(
H^q_\mathrm{{VI}}(\al_0,\al_1,\al_2+\ka,\al_3,\al_4)+A_2+\ka(\al_2+\ka-\ep_2)(x-t)-\ka B_2
\right),\label{eq VI euler q}
\end{align}
where 
\begin{align*}
&A_2=-\ep_2((\al_3+\al_1-\ka+(\al_0+\al_1-\ka)t-(\al_2+\ka)(\al_3+\al_1+\al_2+\ep_1+(\al_0+\al_1+\al_2+\ep_2)t),
\\
&B_2=\al_1(t+1-x)+\al_0t+\al_3+\al_2 x+\ka(2x+t+1).
\end{align*}

We substitute \eqref{eq VI euler x} and \eqref{eq VI euler q} into the left hand side of \eqref{eq VI euler2} and then we 
compute it  directly by using the commutation relations.  After  straightforward calculations, 
we obtain the relation \eqref{eq VI euler2}. 
\end{proof}

%%%%%%%%%%%%%%%%%%%%%%%%%%%%%%%%%%%%%%%%%%%%%%
%%%%%%%%%%%%%%%%%%%%%%%%%%%%%%%%%%%%%%%%%%%%%%
\subsection{$\mathrm{P}_\mathrm{{V}}$ case}

Let $\mathcal{K}$ be the skew field over $\mathbb{C}$ defined by 
the generators $x$, $y$, $q$, $p$, $t$, $d$, $\alpha_i$ ($0\le i\le 3$), $\ep_1,\ep_2$, 
and 
the commutation relations: 
\begin{align*}
&[y,x]=\epsilon_1,\quad [p,q]=\epsilon_2,\quad [d,t]=1, 
\end{align*}
and 
the other commutation relations are zero, and a relation $\al_0+\al_1+\al_2+\al_3=-\ep_1+\ep_2$.

Let $H^x_\mathrm{{V}}(\al)$  ($\al=(\al_0,\al_1,\al_2,\al_3)$) 
be the  Hamiltonian for  the quantum fifth  Painlev\'e  equation defined by 
\begin{equation*}
H^x_\mathrm{{V}}(\al)=(x-1)(y+t)xy-(\al_1+\al_3-\ep_1)xy+\al_1y+(\al_2+\ep_1)tx. 
\end{equation*}
Let $H^q_\mathrm{{V}}(\al)$ be defined by replacing 
$x$, $y$, $\ep_1$ in $H^x_\mathrm{{V}}(\al)$ with $q$, $p$, $\ep_2$, respectively. 

Let us introduce the quantum Lax operators $L_\mathrm{{V}}(\al)$ and $B_V(\al)$ for the fifth Painlev\'e equation defined by 
\begin{align*}
&L_\mathrm{{V}}(\al)=H^x_\mathrm{{V}}(\al_0,\al_1,\al_2,\al_3)-H^q_\mathrm{{V}}(\al_0+\ka,\al_1,\al_2+\ka,\al_3)-{\ka \over x-q}(x(x-1)y-q(q-1)p), 
\\
&B_V(\al)=\ep_2H^x_\mathrm{{V}}(\al_0, \al_1,\al_2,\al_3)-\ep_1H^q_\mathrm{{V}}(\al_0+\ka,\al_1,\al_2+\ka,\al_3)-\ka \ep_1\ep_2 t d. 
\end{align*}

Let us recall the extended affine Weyl group $\widetilde{W}(A_3^{(1)})$ symmetry 
of the quantum fifth Painlev\'e equation. 
Here, $\widetilde{W}(A_3^{(1)})=W(A_3^{(1)})\rtimes G$, where  
$W(A_3^{(1)})=\langle s_0,s_1,s_2,s_3\rangle$ is the affine Weyl group of type $A_3^{(1)}$ and $G=\langle \pi, \sigma\rangle$ 
is the  automorphism group of the  Dynkin diagram of type $A_3^{(1)}$. 

\begin{dfn}[cf. \cite{N QNY}]\label{dfn QPV D}
Let the automorphisms $s^q$ for $s\in\{ s_0,s_1,s_2,s_3,\pi, \sigma\}$  
on $\mathcal{K}$ be defined by the following table:  
\medskip
\begin{center}
\begin{tabular}{|c|cccc|cc|cc|}
\hline
$z$&$\alpha_0$&$\alpha_1$&$\alpha_2$&$ \alpha_3$&$q$&$p$&$t$&$d$
\\ \hline
$s^q_0(z)$ & $-\alpha_0$ & $\alpha_1+\alpha_0$ & $\alpha_2$ & $\alpha_3+\alpha_0$
  & $q+{\alpha_0\over {p+t}}$ & $p$&$t$&$d-{\alpha_0/\ep_2\over p+t}$
\\
$s^q_1(z)$&$\alpha_0+\alpha_1$ & $-\alpha_1$ & $\alpha_2+\alpha_1$ & $\alpha_3$
  & $q$ & $p-{\alpha_1\over q}$&$t$&$d$
\\
$s^q_2(z)$&$\alpha_0$ & $\alpha_1+\alpha_2$ & $-\alpha_2$ 
& $\alpha_3+\alpha_2$
 & $q+{\alpha_2\over p}$ & $p$&$t$&$d$
\\
$s^q_3(z)$&$\alpha_0+\alpha_3$ & $\alpha_1$ & $\alpha_2+\alpha_3$ & $-\alpha_3$
 & $q$ & $p-{\alpha_3\over q-1}$&$t$&$d$
\\ \hline
$\pi^q(z)$&$\alpha_1$&$\alpha_2$&$\alpha_3$&$\alpha_0$&$-{p\over t}$&$t(q-1)$&
$t$&$d+{(1-q)\over \ep_2t}p$
\\ \hline
$\sigma^q(z)$&$\alpha_2$&$\alpha_1$&$\alpha_0$&$\alpha_3$&$q$&
$p+t$&$-t$&$-d-q/\ep_2$
\\
 \hline
\end{tabular}
\end{center}
\medskip

\end{dfn}

\begin{dfn}
Let the automorphisms $R^x_{s_i}(\al_i)$ ($i=0,1,2,3$), $R^x_\pi$, $R_\sigma$ on $\mathcal{K}$ be defined by
\begin{align*}
&R_{s_0}^x(\al_0)=\mathcal{L}_x^{-1}\circ \Ad\left(  \left(  x+t \right)^{-{\alpha_0\over \ep_1}} \right)
\circ \mathcal{L}_x, \quad R_{s_1}^x(\al_1)=\Ad\left(  x^{-{\alpha_1\over \ep_1}} \right),
\\
&R_{s_2}^x(\al_2)=\mathcal{L}_x^{-1}\circ \Ad\left( x^{-{\alpha_2\over \ep_1}} \right)
\circ \mathcal{L}_x,\quad R_{s_3}^x(\al_3)=\Ad\left(  \left(  x-1 \right)^{-{\alpha_3\over \ep_1}} \right),
\\
&R^x_{\pi}=\left(  x\mapsto t(x-1) \right)\circ \mathcal{L}_x,
\quad
R_{\sigma}=\left(  t\mapsto -t \right)\circ \Ad\left(  \exp\left(  {xt\over \ep_1} \right) \right)
\circ \Ad\left(  \exp\left(  {qt\over \ep_2} \right) \right). 
\end{align*}
The automorphisms $R^q_{s_i}(\al_i)$ ($i=0,1,2,3$), $R^q_{\pi}$ are defined by replacing $x$,  $\ep_1$ in $R_{s_i}^x(\al_i)$, 
$R^x_{\pi}$  
with $q$,  $\ep_2$, respectively. 
\end{dfn}

\begin{prop}
[\cite{N QNY}, \cite{N Weyl}]\label{prop H_V sym}
The automorphisms $R^x_{s_i}(\al_i)$ ($i=0,1,2,3$), $R^x_\pi$, $R_\sigma$ preserve
 the Hamiltonian 
$H^x_\mathrm{{V}}(\al)$ in the following sense. 
\begin{align*}
&R^x_{s_i}(\al_i)\left(  H^x_V(\al) \right)=H^x_V(s_i(\al))+C_{s_i},
\\
&R^x_{\pi}\left(  H^x_V(\al) \right)=H^x_V(\pi^{-1}(\al))+C_{\pi},\quad 
R_\sigma\left(  H^x_V(\al) \right)=H^x_V(\sigma(\al))+C_{\sigma},
\end{align*}
where 
\begin{align*}
&C_{s_0}=-\al_0(\al_2+2\ep_1-\ep_2)+\ka t{\al_0\over y+t}, 
\\
&C_{s_1}=-\al_1(\al_3+t-\ep_1),
\\
&C_{s_2}=-\al_2(\al_0+2\ep_1-\ep_2+t),
\\
&C_{s_3}=-\al_3(\al_1-\ep_1),
\\
&C_{\pi}=\al_3\ep_1+\al_1(\ep_1-t)-\ka\ep_1(x-1),
\\
&C_{\sigma}=(\al_1-\ep_1+\ka x)t. 
\end{align*}
\end{prop}
By definition, the automorphisms $R_{s_i}^q(\al_i)$ ($i=0,1,2,3$), $R^q_{\pi}$ and $R_{\sigma}$ 
 act the Hamiltonian $H^q_\mathrm{{V}}(\al)$ in the same way above.

\begin{dfn}
Let the automorphisms $R_{s_i}$ ($i=1,3$), $R_{\pi}$, 
$T_{\sigma s_1s_3s_2}$, $T_{s_1s_2s_3\pi^{-1}}$, $S$ on $\mathcal{K}$ be defined by
\begin{align*}
&R_{s_i}=R^x_{s_i}(\al_i)R^q_{s_i}(\al_i),\quad R_{\pi}=R^x_{\pi}R^q_{\pi}, 
\\
&T_{\sigma s_1s_3s_2}=R^x_{s_2}(-\al_2-\ka)R^x_{s_1}(\al_1)R^x_{s_3}(\al_3)R_\sigma
R^q_{s_1}(\bar{\al_1})R^q_{s_3}(\bar{\al_3})R^q_{s_2}(\al_2+\ka),
\\
&T_{s_1s_2s_3\pi^{-1}}=R^x_\pi R^x_{s_3}(s_1s_2(\al_3))R^x_{s_2}(s_1(\al_2))R^x_{s_1}(\al_1)
R^q_{s_1}(-\al_1+\ka)R^q_{s_2}(-s_1(\al_2))R^q_{s_3}(-s_1s_2(\al_3))\left(R^q_{\pi}\right)^{-1},
\\
&S=\Ad(D(\al_2)^{-1})R^x_{s_2}(\al_2)R^q_{s_2}(\al_2+\ka), 
\end{align*}
where $\bar{\al_i}=-\sigma s_1s_3s_2(\al_i)=\al_i+\al_2+\ka$  for  $i=1,3$, and $D(\al_2)$ is given in \eqref{eq D}. 
\end{dfn}

\begin{thm}\label{thm V}
The automorphisms $R_{s_i}$ ($i=1,3$), $R_\sigma$, $R^2_\pi$, $T_{\sigma s_1s_3s_2}$, $T_{s_1s_2s_3\pi^{-1}}$
  and $S$ act the quantum Lax operators 
$L_\mathrm{{V}}(\al)$ and $B_\mathrm{{V}}(\al)$ as follows. 

For the automorphisms $R_s$ ($s\in\{ s_1,s_3,\sigma \}$),  
\begin{equation*}
R_{s}\left(
L_\mathrm{{V}}(\al),B_\mathrm{{V}}(\al)
\right)=\left(
L_\mathrm{{V}}(s(\al)),B_\mathrm{{V}}(s(\al))+f_{s}
\right), 
\end{equation*}
where
\begin{equation*}
f_{s_1}=\ka\al_1(\al_3+t),\quad 
f_{s_3}=\ka\al_1\al_3,
\quad
f_{\sigma}=-\ka\al_1 t. 
\end{equation*}

For the automorphism $R^2_\pi$, 
\begin{equation*}
R^2_\pi \left(
(x-q)L_\mathrm{{V}}(\al),B_\mathrm{{V}}(\al)
\right)=\left(
(q-x)L_\mathrm{{V}}(\pi^2(\al)),B_\mathrm{{V}}(\pi^2(\al))-\ka(\al_2+\al_3+\ka)t
\right). 
\end{equation*}

For the automorphisms $T_r$ ($r\in\{ \sigma s_1s_3s_2,s_1s_2s_3\pi^{-1}\}$), 
\begin{align}
&l_{T_{r}} T_{r}\left((x-q)
L_\mathrm{{V}}(\al)
\right)=(x-q)L_\mathrm{{V}}\left(r
(\al)
\right)
,\label{eq V euler}
\\
&T_{r}\left(
B_\mathrm{{V}}(\al)
\right)=B_\mathrm{{V}}\left(r
(\al)
\right)+f_{T_{r}},\nn
\end{align}
 where 
\begin{align*}
%&l_{T_{\sigma s_1s_3s_2}}=(t\mapsto t) \Ad\left(  \exp\left(  {qt\over \ep_2} \right) \right) R^q_{s_1}(\bar{\al_1})R^q_{s_3}(\bar{\al_3})R^q_{s_2}(\al_2+\ka)
%\left(\left(\left(
%(x-q)y+\al_2+\ka
%\right)
%p\right)^{-1}
%\left(
%(x-q)p-\al_2-\ka
%\right)
%y\right),
%\\
&f_{T_{\sigma s_1s_3s_2}}=\ka(\al_2-\ka)(\al_0-t), \quad f_{T_{s_1s_2s_3\pi^{-1}}} =0, 
\end{align*}
and $l_{T_{\sigma s_1s_3s_2}}$, $l_{T_{s_1s_2s_3\pi^{-1}}}$ are  some elements in $\mathcal{K}$ whose explicit forms are given in the proof.

For the automorphism $S$, 
\begin{align*}
&yp\left(R^x_{s_2}(\al_2)R^q_{s_2}(\al_2+\ka)\left((x-q)L_\mathrm{{V}}(\al)\right)\right)D(\al_2)
\\
&=\left(
(x-q)yp+(\al_2+\ka-\ep_2)y+(\ep_1-\al_2)p
\right)
\left(
D(\al_2)-{(\al_2+\ep_1)(\ep_1+\ep_2)\over (x-q)^2}
\right)L_\mathrm{{V}}(\al_0,\tilde{\al_1},-\al_2-2\ep_1, \tilde{\al_3}), 
\\
&S\left(
B_\mathrm{{V}}(\al)
+f_S
\right)=B_\mathrm{{V}}(\al_0,\tilde{\al_1},-\al_2-2\ep_1, \tilde{\al_3})
-D(\al_2)^{-1}{2(\al_2+\ep_1)\ep_1\ep_2\over (x-q)^2}L_\mathrm{{V}}(\al_0,\tilde{\al_1},-\al_2-2\ep_1, \tilde{\al_3}), 
\end{align*}
where $\tilde{\al_i}=\al_i+\al_2+\ep_1$ ($i=1,3$) and 
\begin{equation*}
f_S=\ka(\al_2+\ep_1)(\al_1+\al_2+\al_3+\ep_1+t). 
\end{equation*}
\end{thm}
\begin{proof}
For the cases of
the automorphisms $T_r$ ($r\in\{ \sigma s_1s_3s_2, s_1s_2s_3\pi^{-1}\}$ acting $L_\mathrm{V}(\al)$, we show that 
\begin{align}
&\left(
(x-q)p-\al_2-\ka
\right)
y
R_{s_2}^x(-\al_2-\ka)R^x_{s_1}(\al_1)R^x_{s_3}(\al_3)\Ad\left(  \exp\left(  -{xt\over \ep_1} \right) \right)
\left((x-q)
L_\mathrm{{V}}(\al)
\right)\nn
\\
&=\left(
(x-q)y+\al_2+\ka
\right)
p
R^2_{s_2}(-\al_2-\ka)R^q_{s_1}(\bar{\al_1})R^q_{s_3}(\bar{\al_3})\nn
\\
&\circ
(t\mapsto -t) \Ad\left(  \exp\left(  -{qt\over \ep_2} \right) \right)
\left((x-q)
L_\mathrm{{V}}(\sigma s_1s_3s_2
(\al)
\right),
\label{eq V euler2} 
\\
&AR^x_\pi \left((x-1)R^x_{s_3}(s_1s_2(\al_3))\left(yR^x_{s_2}(s_1(\al_2))R^x_{s_1}(\al_1)\left(
(x-q) L_\mathrm{V}(\al)
\right)\right)\right)\nn
\\
&=BR^q_{\pi}\left((q-1)R^q_{s_3}(s_1s_2(\al_3))\left(pR^q_{s_2}(s_1(\al_2))
R^q_{s_1}(\al_1-\ka)\left((x-q)L_\mathrm{V}(s_1s_2s_3\pi^{-1}(\al)) \right)\right)\right), \label{eq V euler3} 
\end{align}
which are the explicit forms of \eqref{eq V euler}. Here $A$, $B$ are elements in $\mathcal{K}$ such that 
\begin{align*}
&A=a_{0,3}p^3-{q-1\over x-1}yp^2+a_{0,2}p^2+(\al_3-\ep_2+t(1-q)(1+x))yp+a_{1,0}y+a_{0,1}p+a_{0,0},
\\
&B=y^3+b_{2,1}y^2p+b_{2,0}y^2+b_{1,1}yp+b_{1,0}y+b_{0,1}p+b_{0,0}. 
\end{align*}
where $a_{i,j}$, $b_{i,j}$ are rational functions of $x$, $q$, $t$, $\al_i$ ($i=1,2,3$), $\ep_1$, $\ep_2$. 
We omit the proofs of  \eqref{eq V euler2}, \eqref{eq V euler3}, since they are similar to that of Theorem \ref{thm VI}.  

Proofs of the other cases
 follow from 
direct computations by using Proposition \ref{prop H_V sym}. 
\end{proof}

Actions involving $R^x_{s_0}(\al_0)$, $R^q_{s_0}(\al_0)$ on the quantum Lax operators can be 
obtained from Theorem \ref{thm V}, because of the relations 
\begin{equation*}
R_\sigma R^x_{s_0}(\al_2) R_\sigma=R^x_{s_2}(\al_2),\quad R_\sigma R^q_{s_0}(\al_2) R_\sigma=R^q_{s_2}(\al_2). 
\end{equation*}

%%%%%%%%%%%%%%%%%%%%%%%%%%%%%%%%%%%%%%%%%%%%%%
\subsection{$\mathrm{P}_\mathrm{{IV}}$ case}

Let $\mathcal{K}$ be the skew field over $\mathbb{C}$ defined by 
the generators $x$, $y$, $q$, $p$, $t$, $d$, $\alpha_i$ ($0\le i\le 2$), $\ep_1,\ep_2$, 
and 
the commutation relations: 
\begin{align*}
&[y,x]=\epsilon_1,\quad [p,q]=\epsilon_2,\quad [d,t]=1, 
\end{align*}
and 
the other commutation relations are zero, and a relation $\al_0+\al_1+\al_2=-\ep_1+\ep_2$. 

Let $H^x_\mathrm{{IV}}(\al)$  ($\al=(\al_0,\al_1,\al_2)$) 
be the  Hamiltonian for  the quantum fourth Painlev\'e  equation defined by 
\begin{equation*}
H^x_\mathrm{{IV}}(\al)=yxy-xyx-txy-\al_2x-\al_1y.
\end{equation*}
Let $H^q_\mathrm{{IV}}(\al)$ be defined by replacing 
$x$, $y$, $\ep_1$, $\ep_2$ in $H^x_\mathrm{{IV}}(\al)$ with $q$, $p$, $\ep_2$, $\ep_1$, respectively. 

Let us introduce the quantum Lax operators  $L_\mathrm{{IV}}(\al)$ and $B_\mathrm{{IV}}(\al)$ for the fourth Painlev\'e equation defined by 
\begin{align*}
&L_\mathrm{{IV}}(\al)=H^x_\mathrm{{IV}}(\al_0,\al_1,\al_2)-H^q_\mathrm{{IV}}(\al_0+\ka,\al_1,\al_2+\ka)-{\ka \over x-q}(xy-qp), 
\\
&B_\mathrm{{IV}}(\al)=\ep_2H^x_\mathrm{{IV}}(\al_0,\al_1,\al_2)-\ep_1H^q_\mathrm{{IV}}(\al_0+\ka,\al_1,\al_2+\ka)-\ka \ep_1\ep_2 d. 
 \end{align*}

Let us recall  the extended affine Weyl group 
$\widetilde{W}(A_2^{(1)})$ symmetry of the quantum fourth Painlev\'e equation. 
Here, $\widetilde{W}(A_2^{(1)})=W(A_2^{(1)})\rtimes G$, 
where $W(A_2^{(1)})=\langle s_0,s_1,s_2\rangle$ is the affine Weyl group of type $A_2^{(1)}$ and $G=\langle \pi,\sigma\rangle$ 
is the automorphism group of the  Dynkin diagram of type $A_2^{(1)}$. 

\begin{dfn}[cf. \cite{N QNY}]\label{dfn QPIV D}
Let the automorphisms $s^q$ for $s\in\{ s_0,s_1,s_2, \pi, \sigma\}$   on $\mathcal{K}$ be defined by the following table:  
\begin{center}
\begin{tabular}{|c|ccc|cc|cc|}
\hline
$z$&$\alpha_0$&$\alpha_1$&$\alpha_2$&$q$&$p$&$t$ &$d$
\\ \hline
$s^q_0(z)$ & $-\alpha_0$ & $\alpha_1+\alpha_0$ &  $\alpha_2+\alpha_0$
  & $q+{\alpha_0\over {p-q-t}}$ & $p+{\alpha_0\over p-q-t}$&$t$&$d+{\alpha_0/\ep_2\over p-q-t}$
\\
$s^q_1(z)$&$\alpha_0+\alpha_1$ & $-\alpha_1$ & $\alpha_2+\alpha_1$ 
  & $q$ & $p-{\alpha_1\over q}$&$t$&$d$
\\
$s^q_2(z)$ & $\alpha_0+\alpha_2$ & $\alpha_1+\alpha_2$ &$-\alpha_2$ 
 & $q+{\alpha_2\over p}$ & $p$&$t$&$d$
\\ \hline
%$\pi_1(x)$ &$\alpha_3$ & $\alpha_4$& $\alpha_2$ &$\alpha_0$ &$\alpha_1$&${t\over q}$
%&$-{q(pq+\alpha_2)\over t}$ &$t$
%\\
%$\pi_2(x)$&$\alpha_1$ &$\alpha_0$ &$\alpha_2$&$\alpha_4$&$\alpha_3$&${(q-1)t\over q-t}$
%&$-{(q-t)(p(q-t)+\alpha_2)\over t(t-1)}$& $t$
%\\ \hline
$\pi^q(z)$&$\alpha_1$&$\alpha_2$&$\alpha_0$&$-p$&$-p+q+t$&
$t$&$d-p$
\\ 
$\sigma^q(z)$&$\alpha_2$&$\alpha_1$&$\alpha_0$&$\sqrt{-1}q$&
$-\sqrt{-1}(p-q-t)$&$\sqrt{-1}t$&$\sqrt{-1}(-d+{q\over \ep_2})$
\\
 \hline
\end{tabular}
\end{center}
\end{dfn}

\begin{dfn}[cf. \cite{N Weyl}]
Let the automorphisms $R^x_{s_i}(\al_i)$ ($i=0,1,2$), $R^x_\pi$, $R_\sigma$ on $\mathcal{K}$ 
be defined by
\begin{align*}
&R^x_{s_0}(\al_0)=\Ad\left( \exp\left(\left({x^2\over 2} +xt\right){1\over \ep_1}\right)\right)
\circ \mathcal{L}_x^{-1}\circ \Ad(x^{-{\alpha_0\over \ep_1}})\circ  \mathcal{L}_x
\circ \Ad\left( \exp\left(\left(-{x^2\over 2} -xt\right){1\over \ep_1}\right)\right),
\\
&R^x_{s_1}(\al_1)=\Ad(x^{-{\alpha_1\over \ep_1}}),\quad 
R^x_{s_2}(\al_2)=\mathcal{L}_x^{-1}\circ \Ad(x^{-{\alpha_2\over \ep_1}})
\circ \mathcal{L}_x,
\\
&R^x_\pi=\mathcal{L}_x\circ \Ad\left( \exp\left(\left(-{x^2\over 2} -xt\right){1\over \ep_1}\right)\right), 
\\
& R_\sigma=\left( x\mapsto \sqrt{-1}x, q\mapsto \sqrt{-1}q, t\mapsto \sqrt{-1}t\right)
\circ \Ad\left( \exp\left(\left(-{x^2\over 2} -xt\right){1\over \ep_1}\right)\right)\circ \Ad\left( \exp\left(\left(-{q^2\over 2} -qt\right){1\over \ep_2}\right)\right).
\end{align*}
The automorphisms $R^q_{s_i}(\al_i)$ ($i=0,1,2$), $R^q_{\pi}$ are defined by replacing $x$, $\ep_1$ in $R_{s_i}^x(\al_i)$, 
$R^x_{\pi}$  
with $q$, $\ep_2$, respectively. 
\end{dfn}

\begin{prop}
[\cite{N QNY}, \cite{N Weyl}]\label{prop H_IV sym}
The automorphisms $R^x_{s_i}(\al_i)$ ($i=0,1,2$), $R^x_\pi$, $R_\sigma$ preserve
 the Hamiltonian 
$H^x_\mathrm{{IV}}(\al)$ in the following sense. 
\begin{align*}
&R^x_{s_i}(\al_i)\left(  H^x_\mathrm{{IV}}(\al) \right)=H^x_\mathrm{{IV}}(s_i(\al))+C_{s_i},
\\
&R^x_{\pi}\left(  H^x_\mathrm{{IV}}(\al) \right)=H^x_\mathrm{{IV}}(\pi^{-1}(\al))+C_{\pi},\quad 
R_\sigma\left(  H^x_\mathrm{{IV}}(\al) \right)=-\sqrt{-1}H^x_\mathrm{{IV}}(\sigma(\al))+C_{\sigma},
\end{align*}
where 
\begin{align*}
&C_{s_0}=-{\ka \al_0\over y-x-t},
\quad C_{s_1}=-\al_1 t,
\quad C_{s_2}=\al_2t,
\\
&C_{\pi}=-\al_1 t-\ka y,
\quad C_{\sigma}=-\sqrt{-1}((\al_1-\ep_1)t-\ka x).  
\end{align*}
\end{prop}
By definition, the automorphisms $R_{s_i}^q(\al_i)$ ($i=0,1,2$), $R^q_{\pi}$ and $R_{\sigma}$ 
 act the Hamiltonian $H^q_\mathrm{{IV}}(\al)$ in the same way above.

\begin{dfn}
Let the automorphisms $R_{s_1}$,  
$T_{\sigma s_1 s_2}$, $T_{s_1 s_2 \pi^{-1}}$, $S$ on $\mathcal{K}$ be defined by
\begin{align*}
&R_{s_1}=R^x_{s_1}(\al_1)R^q_{s_1}(\al_1),\quad 
\\
&T_{\sigma s_1 s_2}=R^x_{s_2}(-\al_2-\ka)R^x_{s_1}(\al_1)R_\sigma
R^q_{s_1}(-\sigma s_1s_2(\al_1))R^q_{s_2}(\al_2+\ka),
\\
&T_{s_1 s_2 \pi^{-1}}=R^x_{\pi}R^x_{s_2}(\al_1+\al_2)R^x_{s_1}(\al_1)R^q_{s_1}(-\al_1+\ka)R^q_{s_2}(\al_0+\ka)\left(R^q_{\pi}\right)^{-1}, 
\\
&S=\Ad(D(\al_2)^{-1})R^x_{s_2}(\al_2)R^q_{s_2}(\al_2+\ka).  
\end{align*}
\end{dfn}

\begin{thm}\label{thm IV}
The automorphisms $R_{s_1}$, $R_\sigma$,  $T_{\sigma s_1 s_2}$, $T_{s_1 s_2 \pi^{-1}}$ and $S$ act the quantum Lax operators 
$L_\mathrm{{IV}}(\al)$ and $B_\mathrm{{IV}}(\al)$ as follows. 

For the automorphisms $R_s$ ($s\in\{ s_1,\sigma \}$),  
\begin{equation*}
R_{s}\left(
L_\mathrm{{IV}}(\al),B_\mathrm{{IV}}(\al)
\right)=c_s\left(
L_\mathrm{{IV}}(s(\al)), B_\mathrm{{IV}}(s(\al))+f_{s}
\right), 
\end{equation*}
where
\begin{equation*}
c_{s_1}=1,\quad c_{\sigma}=-\sqrt{-1},\quad f_{s_1}=\ka\al_1 t,
\quad
f_{\sigma}=-\ka\al_1 t. 
\end{equation*}

For the automorphism $T_r$ ($r\in\{\sigma s_1 s_2, s_1 s_2 \pi^{-1}\}$,  
\begin{align}
&l_{T_{r}} T_{r}\left((x-q)
L_\mathrm{{IV}}(\al)
\right)=(x-q)L_\mathrm{{IV}}\left(r
(\al)
\right)
,\label{eq IV euler}
\\
&T_{r}\left(
B_\mathrm{{IV}}(\al)
\right)=-\sqrt{-1}B_\mathrm{{IV}}\left(r
(\al)
\right)+f_{T_{r}},\nn
\end{align}
 where 
\begin{align*}
%&l_{T_{\sigma s_1 s_2}}= R_\sigma \Ad\left( \exp\left(\left({x^2\over 2} +xt\right){1\over \ep_1}\right)\right)R^q_{s_1}(\bar{\al_1})R^q_{s_2}(\al_2+\ka)
%\left(\left(\left(
%(x-q)y+\al_2+\ka
%\right)
%p\right)^{-1}
%\left(
%(x-q)p-\al_2-\ka
%\right)
%y\right),
%\\
&f_{T_{\sigma s_1 s_2}}=-\ka(\al_2+\ka)t, \quad f_{T_{s_1 s_2 \pi^{-1}}}=0,
\end{align*}
and $l_{T_{\sigma s_1s_2}}$, $l_{T_{s_1s_2\pi^{-1}}}$ are  some elements in $\mathcal{K}$ whose explicit forms are given in the proof.

For the automorphism $S$, 
\begin{align*}
&yp\left(R^x_{s_2}(\al_2)R^q_{s_2}(\al_2+\ka)\left((x-q)L_\mathrm{{IV}}(\al)\right)\right)D(\al_2)
\\
&=\left(
(x-q)yp+(\al_2+\ka-\ep_2)y+(\ep_1-\al_2)p
\right)
\left(
D(\al_2)-{(\al_2+\ep_1)(\ep_1+\ep_2)\over (x-q)^2}
\right)L_\mathrm{{IV}}(\al_0,\tilde{\al_1},-\al_2-2\ep_1), 
\\
&S\left(
B_\mathrm{{IV}}(\al)
+f_S
\right)=B_\mathrm{{IV}}(\al_0,\tilde{\al_1},-\al_2-2\ep_1)
-D(\al_2)^{-1}{2(\al_2+\ep_1)\ep_1\ep_2\over (x-q)^2}L_\mathrm{{IV}}(\al_0,\tilde{\al_1},-\al_2-2\ep_1), 
\end{align*}
where $\tilde{\al_1}=\al_1+\al_2+\ep_1$  and 
\begin{equation*}
f_S=\ka(\al_2+\ep_1) t. 
\end{equation*}
\end{thm}

\begin{proof}
For the cases of
the automorphisms $T_r$ ($r\in\{ \sigma s_1s_2, s_1s_2\pi^{-1}\}$ acting $L_\mathrm{IV}(\al)$, we show that 
\begin{align}
&\left(
(x-q)p-\al_2-\ka
\right)
y
R_{s_2}^x(-\al_2-\ka)R^x_{s_1}(\al_1) \Ad\left( \exp\left(\left(-{x^2\over 2} -xt\right){1\over \ep_1}\right)\right)
\left((x-q)
L_\mathrm{{IV}}(\al)
\right)\nn
\\
&=\left(
(x-q)y+\al_2+\ka
\right)
p
R^2_{s_2}(-\al_2-\ka)R^q_{s_1}(\sigma s_1s_2(\al_1))\nn
\\
&\circ
\left( x\mapsto \sqrt{-1}x, q\mapsto \sqrt{-1}q, t\mapsto \sqrt{-1}t\right)\Ad\left( \exp\left(\left(-{q^2\over 2} -qt\right){1\over \ep_2}\right)\right)
\left((x-q)
L_\mathrm{{IV}}(\sigma s_1s_2
(\al)
\right),
\label{eq IV euler2} 
\\
&AR^x_\pi \left(yR^x_{s_2}(s_1(\al_2))R^x_{s_1}(\al_1)\left(
(x-q) L_\mathrm{IV}(\al)
\right)\right)\nn
\\
&=BR^q_{\pi}\left(pR^q_{s_2}(s_1(\al_2))
R^q_{s_1}(\al_1-\ka)\left((x-q)L_\mathrm{IV}(s_1s_2\pi^{-1}(\al)) \right)\right), \label{eq IV euler3} 
\end{align}
which is the explicit form of \eqref{eq IV euler}. Here $A$, $B$ are elements in $\mathcal{K}$ such that 
\begin{align*}
&A=a_{0,3}p^3-yp^2+a_{0,2}p^2+(q-x+t)yp+(\al_1+\al_2+(q+t)x)y+a_{0,1}p+a_{0,0},
\\
&B=b_{3,0}y^3+b_{2,1}y^2p+b_{2,0}y^2+b_{1,1}yp+b_{1,0}y+b_{0,1}p+b_{0,0}. 
\end{align*}
where $a_{i,j}$, $b_{i,j}$ are rational functions of $x$, $q$, $t$, $\al_i$ ($i=1,2,3$), $\ep_1$, $\ep_2$. 
We omit the proofs of \eqref{eq IV euler2}, \eqref{eq IV euler3}, since they are similar to that of Theorem \ref{thm VI}. 

Proofs of the other cases
 follow from 
direct computations by using Proposition \ref{prop H_IV sym}. 
\end{proof}

Actions involving $R^x_{s_0}(\al_0)$, $R^q_{s_0}(\al_0)$ on the quantum Lax operators can be 
obtained from Theorem \ref{thm IV}, because of the relations 
\begin{equation*}
\left(R_\sigma\right)^{-1} R^x_{s_0}(\al_2) R_\sigma=R^x_{s_2}(\al_2),\quad \left(R_\sigma\right)^{-1} R^q_{s_0}(\al_2) R_\sigma=R^q_{s_2}(\al_2). 
\end{equation*}

%%%%%%%%%%%%%%%%%%%%%%%%%%%%%%%%%%%%%%%%%%%%%%
\subsection{$\mathrm{P}_\mathrm{{III}}$ case}

Let $\mathcal{K}$ be the skew field over $\mathbb{C}$ defined by 
the generators $x$, $y$, $q$, $p$, $t$, $d$, $\alpha_i$ ($0\le i\le 2$), $\ep_1,\ep_2$, 
and 
the commutation relations: 
\begin{align*}
&[y,x]=\epsilon_1,\quad [p,q]=\epsilon_2,\quad [d,t]=1, 
\end{align*}
and 
the other commutation relations are zero, and a relation $\al_0+2\al_1+\al_2=-\ep_1+\ep_2$.

Let $H^x_\mathrm{{III}}(\al)$  ($\al=(\al_0,\al_1,\al_2)$) 
be the  Hamiltonian for  the quantum third Painlev\'e  equation defined by 
\begin{equation*}
H^x_\mathrm{{III}}(\al)=xyxy-xyx+(\alpha_0+\alpha_2+\epsilon_1)xy-\alpha_2 x+ty.
\end{equation*}
Let $H^q_\mathrm{{III}}(\al)$ be defined by replacing 
$x$, $y$, $\ep_1$, $\ep_2$ in $H^x_\mathrm{{III}}(\al)$ with $q$, $p$, $\ep_2$, $\ep_1$, respectively. 

Let us introduce the quantum Lax operators $L_\mathrm{{III}}(\al)$ and $B_\mathrm{{III}}(\al)$ for the third Painlev\'e equation defined by 
\begin{align*}
L_\mathrm{{III}}(\al)=&H^x_\mathrm{{III}}(\al_0,\al_1,\al_2)-H^q_\mathrm{{III}}(\al_0+\ka,\al_1,\al_2+\ka)-{\ka x q\over x-q}(y-p), 
\\
B_\mathrm{{III}}(\al)=&\ep_2H^x_\mathrm{{III}}(\al_0,\al_1,\al_2)-\ep_1H^q_\mathrm{{III}}(\al_0+\ka,\al_1,\al_2+\ka)-\ka\ep_1\ep_2 t d.
\end{align*}

Let us recall  the extended affine Weyl group  
$\widetilde{W}(C_2^{(1)})$ symmetry of the quantum third Painlev\'e equation.  Here, $\widetilde{W}(C_2^{(1)})=W(C_2^{(1)})\rtimes G$, where 
$W(C_2^{(1)})=\langle s_0,s_1,s_2\rangle$ is the affine Weyl group of type $C_2^{(1)}$
 and $G=\langle \sigma \rangle $ is the automorphism group 
of the  Dynkin diagram of type $C_2^{(1)}$. 

\begin{dfn}[cf. \cite{NGR}, \cite{JNS}]\label{dfn QPIII D}
Let the automorphisms $s^q$ for $s\in\{ s_0,s_1,s_2,  \sigma\}$ on $\mathcal{K}$ be defined by the following table:  
\begin{center}
\begin{tabular}{|c|ccc|cc|cc|}
\hline
$z$&$\alpha_0$&$\alpha_1$&$\alpha_2$&$q$&$p$&$z$ &$d$
\\ \hline
$s^q_0(z)$ & $-\alpha_0$ & $\alpha_1+\alpha_0$ &  $\alpha_2$
  & $q+{\alpha_0\over {p-1}}$ & $p$&$t$&$d$
\\
$s^q_1(z)$&$\alpha_0+2\alpha_1$ & $-\alpha_1$ & $\alpha_2+2\alpha_1$ 
  & $q$ & $p-{2\alpha_1\over q}+{t\over q^2}$&$-t$&$-d+{1\over \ep_2 q}$
\\
$s^q_2(z)$ & $\alpha_0$ & $\alpha_1+\alpha_2$ &$-\alpha_2$ 
 & $q+{\alpha_2\over p}$ & $p$&$t$&$d$
\\ \hline
$\sigma^q(z)$&$\alpha_2$&$\alpha_1$&$\alpha_0$&$-q$&
$1-p$&$-t$&$-d$
\\
 \hline
\end{tabular}
\end{center}
\end{dfn}

\begin{dfn}[cf. \cite{N Weyl}]
Let the automorphisms $R^x_{s_i}(\al_i)$ ($i=0,1,2$),  $R_\sigma$ on $\mathcal{K}$ 
be defined by
\begin{align*}
&R^x_{s_0}(\al_0)=\mathcal{L}_x^{-1}\circ \Ad\left(  (x-1)^{-{\alpha_0\over \ep_1}} \right)\circ \mathcal{L}_x,
\\
& R^x_{s_1}(\al_1)=\left(  t\mapsto -t \right)\circ\Ad\left(  \exp\left( - {t\over \ep_1 x} \right)x^{-{2\alpha_1\over \ep_1}} \right)
,
\\
&R^x_{s_2}(\al_2)=\mathcal{L}_x^{-1}\circ \Ad\left(  
x^{-{\alpha_2\over \ep_1}} \right)\circ \mathcal{L}_x,
\\
& R_\sigma=\left(  x\mapsto -x, q\mapsto -q, \ t\mapsto -t \right)\circ 
\Ad\left(  \exp\left(-{x\over \ep_1}\right) \right)\circ\Ad\left(  \exp\left(-{q\over \ep_2}\right) \right).
\end{align*}
The automorphisms $R^q_{s_i}(\al_i)$ ($i=0,1,2$) are defined by replacing $x$, $\ep_1$ in $R_{s_i}^x(\al_i)$, 
with $q$, $\ep_2$, respectively. 
\end{dfn}

\begin{prop}
[\cite{N Weyl}]\label{prop H_III sym}
The automorphisms $R^x_{s_i}(\al_i)$ ($i=0,1,2$),  $R_\sigma$ preserve
 the Hamiltonian 
$H^x_\mathrm{{III}}(\al)$ in the following sense. 
\begin{align*}
&R^x_{s_i}(\al_i)\left(  H^x_\mathrm{{III}}(\al) \right)=H^x_\mathrm{{III}}(s_i(\al))+C_{s_i},
\\ 
&R_\sigma\left(  H^x_\mathrm{{III}}(\al) \right)=H^x_\mathrm{{III}}(\sigma(\al))+C_{\sigma},
\end{align*}
where 
\begin{equation*}
C_{s_0}=-(\al_0+\ep_1)(\al_2+\ep_1),
\quad C_{s_1}=2\al_1\ep_2- t-{\ka t\over x},
\quad C_{s_2}=-\al_2(\al_0+\ep_1),
\quad C_{\sigma}=\ka t.  
\end{equation*}
\end{prop}
By definition, the automorphisms $R_{s_i}^q(\al_i)$ ($i=0,1,2$) and $R_{\sigma}$ 
 act the Hamiltonian $H^q_\mathrm{{III}}(\al)$ in the same way above.

\begin{dfn}
Let the automorphisms $R_{s_1}$,  
$T_{\sigma s_1 s_2}$, $S$ on $\mathcal{K}$ be defined by
\begin{align*}
&R_{s_1}=R^x_{s_1}(\al_1)R^q_{s_1}(\al_1),\quad 
\\
&T_{\sigma s_1 s_2}=R^x_{s_2}(-\al_2-\ka)R^x_{s_1}(\al_1)R_\sigma (t \mapsto t)
R^q_{s_1}(-\sigma s_1 s_2(\al_1))R^q_{s_2}(\al_2+\ka),
\\
&S=\Ad(D(\al_2)^{-1})R^x_{s_2}(\al_2)R^q_{s_2}(\al_2+\ka). 
\end{align*}
\end{dfn}

\begin{thm}\label{thm III}
The automorphisms $R_{s_1}$, $R_\sigma$,  $T_{\sigma s_1 s_2}$ and $S$ act the quantum Lax operators 
$L_\mathrm{{III}}(\al)$ and $B_\mathrm{{III}}(\al)$ as follows. 

For the automorphisms $R_s$ ($s\in\{ s_1,\sigma \}$),  
\begin{equation*}
R_{s}\left(
L_\mathrm{{III}}(\al),B_\mathrm{{III}}(\al)
\right)=\left(
L_\mathrm{{III}}(s(\al)), B_\mathrm{{III}}(s(\al))+f_{s}
\right), 
\end{equation*}
where
\begin{equation*}
 f_{s_1}=\ka( t-2\al_1(\ep_1+\ep_2)),
\quad
f_{\sigma}=\ka t. 
\end{equation*}

For the automorphism $T_{\sigma s_1 s_2}$,  
\begin{align}
&l_{T_{\sigma s_1 s_2}} T_{\sigma s_1 s_2}\left((x-q)
L_\mathrm{{III}}(\al)
\right)=(x-q)L_\mathrm{{III}}\left(\sigma s_1s_2
(\al)
\right)
,\label{eq III euler}
\\
&T_{\sigma s_1 s_2}\left(
B_\mathrm{{III}}(\al)
\right)=B_\mathrm{{III}}\left(\sigma s_1s_2
(\al)
\right)+f_{T_{\sigma s_1 s_2}},\nn
\end{align}
 where 
\begin{align*}
%&l_{T_{\sigma s_1 s_2}}= R_\sigma \Ad\left( \exp\left({x\over \ep_1}\right)\right)R^q_{s_1}(\bar{\al_1})R^q_{s_2}(\al_2+\ka)
%\left(\left(\left(
%(x-q)y+\al_2+\ka
%\right)
%p\right)^{-1}
%\left(
%(x-q)p-\al_2-\ka
%\right)
%y\right),
%\\
&f_{T_{\sigma s_1 s_2}}=\ka\al_0(\al_2+2\ep_1). 
\end{align*}
and $l_{T_{\sigma s_1s_2}}$ is  some element in $\mathcal{K}$ whose explicit form is given in the proof.

For the automorphism $S$, 
\begin{align*}
&yp\left(R^x_{s_2}(\al_2)R^q_{s_2}(\al_2+\ka)\left((x-q)L_\mathrm{{III}}(\al)\right)\right)D(\al_2)
\\
&=\left(
(x-q)yp+(\al_2+\ka-\ep_2)y+(\ep_1-\al_2)p
\right)
\left(
D(\al_2)-{(\al_2+\ep_1)(\ep_1+\ep_2)\over (x-q)^2}
\right)L_\mathrm{{III}}(\al_0,\tilde{\al_1},-\al_2-2\ep_1), 
\\
&S\left(
B_\mathrm{{III}}(\al)
+f_S
\right)=B_\mathrm{{III}}(\al_0,\tilde{\al_1},-\al_2-2\ep_1)
-D(\al_2)^{-1}{2(\al_2+\ep_1)\ep_1\ep_2\over (x-q)^2}L_\mathrm{{III}}(\al_0,\tilde{\al_1},-\al_2-2\ep_1), 
\end{align*}
where $\tilde{\al_1}=\al_1+\al_2+\ep_1$  and 
\begin{equation*}
f_S=-\ka(\al_0+\ep_1)(\al_2+\ep_1) . 
\end{equation*}

\end{thm}

\begin{proof}
For the cases of
the automorphisms $T_{\sigma s_1s_2}$  acting $L_\mathrm{III}(\al)$, we show that 
\begin{align}
&\left(
(x-q)p-\al_2-\ka
\right)
y
R_{s_2}^x(-\al_2-\ka)R^x_{s_1}(\al_1) \Ad\left(\exp\left(\left(-{t\over x} -x\right){1\over \ep_1}\right)
x^{-{2\alpha_1\over \ep_1}}\right)
\left((x-q)
L_\mathrm{{III}}(\al)
\right)\nn
\\
&=-\left(
(x-q)y+\al_2+\ka
\right)
p
R^2_{s_2}(-\al_2-\ka)R^q_{s_1}(\sigma s_1s_2(\al_1))\nn
\\
&\circ
\left( x\mapsto -x, q\mapsto -q\right)\Ad\left( \exp\left(\left(-{t\over q} -q\right){1\over \ep_2}\right)
q^{-{2\sigma s_1s_2(\al_1)\over \ep_2}}\right)
\left((x-q)
L_\mathrm{{III}}(\sigma s_1s_2
(\al)
\right),  
\label{eq III euler2} 
\end{align}
which is the explicit form of \eqref{eq III euler}. 
We omit the proofs of \eqref{eq III euler2}, since they are similar to that of Theorem \ref{thm VI}.   

Proofs of the other cases
 follow from 
direct computations by using Proposition \ref{prop H_III sym}. 
\end{proof}

Actions involving $R^x_{s_0}(\al_0)$, $R^q_{s_0}(\al_0)$ on the quantum Lax operators can be 
obtained from Theorem \ref{thm III}, because of the relations 
\begin{equation*}
R_\sigma R^x_{s_0}(\al_2) R_\sigma=R^x_{s_2}(\al_2),\quad R_\sigma R^q_{s_0}(\al_2) R_\sigma=R^q_{s_2}(\al_2). 
\end{equation*}

%%%%%%%%%%%%%%%%%%%%%%%%%%%%%%%%%%%%%%%%%%%%%%
\subsection{$\mathrm{P}_\mathrm{{III}}^{D_7}$ case}

Let $\mathcal{K}$ be the skew field over $\mathbb{C}$ defined by 
the generators $x$, $y$, $q$, $p$, $t$, $d$, $\alpha_0$, $\al_1$, $\ep_1,\ep_2$, 
and 
the commutation relations: 
\begin{align*}
&[y,x]=\epsilon_1,\quad [p,q]=\epsilon_2,\quad [d,t]=1, 
\end{align*}
and 
the other commutation relations are zero, and a relation $\al_0+\al_1=-\ep_1+\ep_2$.

Let $H^{D_7,x}_\mathrm{{III}}(\al)$  ($\al=(\al_0,\al_1)$) 
be the  Hamiltonian for  the quantum third Painlev\'e  equation of type $D_7$ defined by 
\begin{equation*}
H^{D_7,x}_\mathrm{{III}}(\al)=xyxy  +  (-\al_0  + \ep_2) xy+ 
    ty +  x.
\end{equation*}
Let $H^{D_7,q}_\mathrm{{III}}(\al)$ be defined by replacing 
$x$, $y$, $\ep_1$, $\ep_2$ in $H^{D_7,x}_\mathrm{{III}}(\al)$ with $q$, $p$, $\ep_2$, $\ep_1$, respectively. 

Let us introduce the quantum Lax operators $L_\mathrm{{III}}^{D_7}(\al)$ and $B_\mathrm{{III}}^{D_7}(\al)$ for the third Painlev\'e equation 
of type $D_7$ 
defined by 
\begin{align*}
L_\mathrm{{III}}^{D_7}(\al)=&H^{D_7,x}_\mathrm{{III}}(\al_0,\al_1)-H^{D_7,q}_\mathrm{{III}}(\al_0,\al_1+2\ka)-{\ka x q\over x-q}(y-p),
\\
B_\mathrm{{III}}^{D_7}(\al)=&\ep_2H^{D_7,x}_\mathrm{{III}}(\al_0,\al_1)-\ep_1H^{D_7,q}_\mathrm{{III}}(\al_0,\al_1+2\ka)-\ka\ep_1\ep_2 td.  
\end{align*}

We introduce  the extended affine Weyl group 
$\widetilde{W}(A_1^{(1)})$ symmetry of the quantum third Painlev\'e equation of type $D_7$. 
 Here, $\widetilde{W}(A_1^{(1)})=W(A_1^{(1)})\rtimes G$, where 
$W(A_1^{(1)})=\langle s_0,s_1\rangle$ is the affine Weyl group of type $A_1^{(1)}$ and $G=\langle \pi\rangle$ is 
the automorphism group of the  Dynkin diagram of type $A_1^{(1)}$. 

\begin{dfn}\label{dfn QPIII D7}
Let the automorphisms $s^q$ for $s\in\{ s_0,s_1,  \pi\}$  on $\mathcal{K}$ be defined by the following table:  
\begin{center}
\begin{tabular}{|c|cc|cc|cc|}
\hline
$z$&$\alpha_0$&$\alpha_1$&$q$&$p$&$t$ &$d$
\\ \hline
$s^q_0(z)$ & $-\alpha_0$ & $\alpha_1+2\alpha_0$ 
  & $q$ & 
  $p-{\alpha_0\over q}+{t\over q^2}$&$-t$&$-d+{1\over \ep_2 q}$
\\
$s^q_1(z)$&$\alpha_0+2\alpha_1$ & $-\alpha_1$  
  & $-q-{\alpha_1\over p}-{1\over p^2}$ & $-p$&$-t$&$-d$
\\ \hline
$\pi^q(z)$&$\alpha_1$&$\alpha_0$&$tp$&$-{q\over t}$&
$-t$&$-d-{q p\over \ep_2 t}$
 \\ \hline
\end{tabular}
\end{center}
\end{dfn}

\begin{dfn}
Let the automorphisms $R^x_{s_i}(\al_i)$ ($i=0,1$), $R^x_\pi$ on $\mathcal{K}$ 
be defined by
\begin{align*}
&R^x_{s_0}(\al_0)=(t\mapsto -t)\circ \Ad\left(  \exp\left(  -{t\over \ep_1 x} \right)x^{-{\al_0\over \ep_1}} \right),
\\
&R^x_{\pi}=\left(x\mapsto -{x\over t}, t\mapsto -t\right) \circ \mathcal{L}_x,
\\
&R^x_{s_1}(\al_1)=R^x_\pi\circ R^x_{s_0}(\al_1)\circ R^x_\pi. 
\end{align*}
The automorphisms $R^q_{s_i}(\al_i)$ ($i=0,1$) and $R^q_\pi$ are defined by replacing $x$, $\ep_1$ in $R_{s_i}^x(\al_i)$, $R^x_{\pi}$ 
with $q$, $\ep_2$, respectively. 
\end{dfn}

\begin{prop}\label{prop H_III D7 sym}
The automorphisms $R^x_{s_0}(\al_0)$,  $R_\pi$ preserve
 the Hamiltonian 
$H^{D_7,x}_\mathrm{{III}}(\al)$ in the following sense. 
\begin{align*}
&R^x_{s_0}(\al_0)\left(  H^{D_7,x}_\mathrm{{III}}(\al) \right)=H^{D_7,x}_\mathrm{{III}}(s_0(\al))+C_{s_0},
\\ 
&R_\pi\left(  H^{D_7,x}_\mathrm{{III}}(\al) \right)=H^{D_7,x}_\mathrm{{III}}(\pi(\al))+C_{\pi},
\end{align*}
where 
\begin{equation*}
C_{s_0}=\ep_2\al_0-{\ka t\over x},
\quad C_{\pi}=-\ep_1 \al_1+\ka x y.  
\end{equation*}
\end{prop}
By definition, the automorphisms $R_{s_0}^q(\al_0)$ and $R_{\pi}$ 
 act the Hamiltonian $H^{D_7,q}_\mathrm{{III}}(\al)$ in the same way above.

\begin{dfn}
Let the automorphisms  $R_{s_0}$,  
$T_{s_0 \pi }$, $S$ on $\mathcal{K}$ be defined by
\begin{align*}
&R_{s_0}=R^x_{s_0}(\al_0)R^q_{s_0}(\al_0), 
\\
&T_{s_0 \pi  }=R^q_{s_0}(-\al_0+\ka)\left(R^q_\pi\right)^{-1}R^x_\pi R^x_{s_0}(\al_0),
\\
&S=\Ad(D_\mathrm{{III}}^{-1})R^q_{\pi}R^q_{s_0}(\al_0), 
\end{align*}
where 
\begin{equation*}
D_\mathrm{{III}}=
{x\over x-q}y+{q\over q-x}p. 
\end{equation*}

\end{dfn}

\begin{thm}\label{thm III D7}
The automorphisms $R_{s_0}$,   $T_{s_0 \pi }$ and $S$ act the quantum Lax operators 
$L_\mathrm{{III}}^{D_7}(\al)$ and $B_\mathrm{{III}}^{D_7}(\al)$ as follows. 

For the automorphisms $R_{s_0}$,   
\begin{equation*}
R_{s_0}\left(
L_\mathrm{{III}}^{D_7}(\al),B_\mathrm{{III}}^{D_7}(\al)
\right)=\left(
L_\mathrm{{III}}^{D_7}(s_0(\al)), B_\mathrm{{III}}^{D_7}(s_0(\al))-\ka\al_0(\ep_1+\ep_2)
\right). 
\end{equation*}

For the automorphism $T_{s_0 \pi }$,  
\begin{align}
&l_{T_{s_0 \pi }} T_{s_0 \pi}\left((x-q)
L_\mathrm{{III}}^{D_7}(\al)
\right)=(x-q)L_\mathrm{{III}}^{D_7}\left(s_0 \pi
(\al)
\right)
,\label{eq III D7 euler}
\\
&T_{s_0 \pi }\left(
B_\mathrm{{III}}^{D_7}(\al)
\right)=B_\mathrm{{III}}^{D_7}\left(s_0 \pi
(\al)
\right)+(\al_0-\ep_1)\ka^2,\nn
\end{align}
 where  $l_{T_{s_0\pi }}$ is  some element in $\mathcal{K}$ whose explicit form is given in the proof. 
%\begin{align*}
%&l_{T_{s_0 \pi }}= R^q_{s_0}(-\al_0+\ka)\left(R^q_\pi\right)^{-1}
%\left(\left(
%ty-q
%\right)^{-1}
%\left(
%-tp+x
%\right)
%\right). 
%\end{align*}

For the automorphism $S$, 
\begin{align*}
&\left(R^q_{\pi}R^q_{s_0}(\al_0)\left((x-q)L_\mathrm{{III}}^{D_7}(\al)\right)\right)D_\mathrm{{III}}
\\
&=\left(
{1\over x-q}\left(
x+{t \ep_2\over q-x}
-tp
\right)
\left(
{\ep_1q+\ep_2x\over x-q}+qp-xy
\right)
+{\ka x\over x-q}
\right)L_\mathrm{{III}}^{D_7}(\al_0+\ep_1,\al_1-\ep_1),
\\
&S\left(
B_\mathrm{{III}}^{D_7}(\al)
\right)=\left(
{x\over x-q}y+{q\over q-x}p
\right)B_\mathrm{{III}}^{D_7}(\al_0+\ep_1,\al_1-\ep_1)-\ep_1\ep_2 D_\mathrm{{III}}^{-1}{x+q\over (x-q)^2}L_\mathrm{{III}}^{D_7}(\al_0+\ep_1,\al_1-\ep_1). 
\end{align*}

\end{thm}

\begin{proof}
For the cases of
the automorphisms $T_{s_0\pi}$  acting $L_\mathrm{III}^{D_7}(\al)$, we show that 
\begin{align}
&\left(
-tp+x
\right)
R^x_\pi R^x_{s_0}(\al_0) \left((x-q)
L_\mathrm{III}^{D_7}(\al)
\right)\nn
\\
&=-\left(
ty-q
\right)
R^q_\pi R^q_{s_0}(\al_0-\ka)
\left((x-q)
L_\mathrm{III}^{D_7}(s_0\pi
(\al)
\right), 
\label{eq III D7 euler2} 
\end{align}
which is the explicit form of \eqref{eq III D7 euler}. 
We omit the proofs of \eqref{eq III D7 euler2}, since they are similar to that of Theorem \ref{thm VI}.   

Proofs of the other cases
 follow from 
direct computations by using Proposition \ref{prop H_III D7 sym}. 
\end{proof}

Actions involving $R^x_{s_1}(\al_1)$, $R^q_{s_1}(\al_1)$ on the quantum Lax operators can be 
obtained from Theorem \ref{thm III D7}, because of the definitions of $R^x_{s_1}(\al_1)$, $R^q_{s_1}(\al_1)$.

%%%%%%%%%%%%%%%%%%%%%%%%%%%%%%%%%%%%%%%%%%%%%%
\subsection{$\mathrm{P_\mathrm{II}}$ case}

Let $\mathcal{K}$ be the skew field over $\mathbb{C}$ defined by 
the generators $x$, $y$, $q$, $p$, $t$, $d$, $\alpha_0$, $\al_1$, $\ep_1,\ep_2$, 
and 
the commutation relations: 
\begin{align*}
&[y,x]=\epsilon_1,\quad [p,q]=\epsilon_2,\quad [d,t]=1, 
\end{align*}
and 
the other commutation relations are zero, and a relation $\al_0+\al_1=-\ep_1+\ep_2$.

Let $H^x_\mathrm{II}(\al)$  ($\al=(\al_0,\al_1)$) 
be the  Hamiltonians for  the quantum second Painlev\'e  equation defined by 
\begin{equation*}
H^x_\mathrm{II}(\al)={y^2\over 2}-xyx-{t\over 2}y-\al_1 x.
\end{equation*}
Let $H^q_\mathrm{II}(\al)$ be defined by replacing 
$x$, $y$, $\ep_1$, $\ep_2$ in $H^x_\mathrm{II}(\al)$ with $q$, $p$, $\ep_2$, $\ep_1$, respectively. 

Let us introduce the quantum Lax operators $L_\mathrm{II}(\al)$ and $B_\mathrm{II}(\al)$ for the second Painlev\'e equation defined by 
\begin{align*}
L_\mathrm{II}(\al)=&H^x_\mathrm{II}(\al_0,\al_1)-H^q_\mathrm{II}(\al_0+\ka,\al_1+\ka)-{\ka \over 2(x-q)}(y-p),%\label{eq Lax L II}
\\
B_\mathrm{II}(\al)=&\ep_2H^x_\mathrm{II}(\al_0,\al_1)-\ep_1H^q_\mathrm{II}(\al_0+\ka,\al_1+\ka)-\ka\ep_1\ep_2 d. %\label{eq Lax B II}
\end{align*}

Let us recall  the extended affine Weyl group 
$\widetilde{W}(A_1^{(1)})$ symmetry of the quantum second Painlev\'e equation. 
 Here, $\widetilde{W}(A_1^{(1)})=W(A_1^{(1)})\rtimes G$, where 
$W(A_1^{(1)})=\langle s_0,s_1\rangle$ is the affine Weyl group of type $A_1^{(1)}$ and $G=\langle \pi\rangle$ is 
the automorphism group of the  Dynkin diagram of type $A_1^{(1)}$. 

\begin{dfn}[cf. \cite{N QNY}, \cite{N Weyl}]\label{dfn QPII D}
Let the automorphisms $s^q$ for $s\in\{ s_0,s_1,  \pi\}$ on $\mathcal{K}$ be defined by the following table:  
\begin{center}
\begin{tabular}{|c|cc|cc|cc|}
\hline
$z$&$\alpha_0$&$\alpha_1$&$q$&$p$&$t$ &$d$
\\ \hline
$s^q_0(z)$ & $-\alpha_0$ & $\alpha_1+2\alpha_0$ 
  & $q+{\alpha_0\over f}$ & 
  $p+2q{\alpha_0\over f}+2{\alpha_0\over f}q+2{\alpha_0^2\over f^2}$&$t$&$d+{\al_0/\ep_2\over f}$
\\
$s^q_1(z)$&$\alpha_0+2\alpha_1$ & $-\alpha_1$  
  & $q+{\alpha_1\over p}$ & $p$&$t$&$d$
\\ \hline
$\pi^q(z)$&$\alpha_1$&$\alpha_0$&$-q$&$-f$&
$t$&$d-{q\over \ep_2}$
 \\ \hline
\end{tabular}
\end{center}
\medskip
where $f=p-2q^2-t$. 
\end{dfn}

\begin{dfn}[cf. \cite{N Weyl}]
Let the automorphisms $R^x_{s_i}(\al_i)$ ($i=0,1$), $R_\pi$ on $\mathcal{K}$ 
be defined by
\begin{align*}
&R^x_{s_1}(\al_1)=\mathcal{L}_x^{-1}\circ \Ad(x^{-{\alpha_1\over \ep_1}})\circ \mathcal{L}_x,
\\
&R_{\pi}=\left( x\mapsto -x,q\mapsto -q\right)\circ \Ad\left(\exp\left(\left( -{2\over 3}x^3-xt\right){1\over \ep_1}\right)\right)
\circ \Ad\left(\exp\left(\left( -{2\over 3}q^3-qt\right){1\over \ep_2}\right)\right), 
\\
&R^x_{s_0}(\al_0)=R_\pi\circ R^x_{s_1}(\al_0)\circ R^x_\pi. 
\end{align*}
The automorphisms $R^q_{s_i}(\al_i)$ ($i=0,1,2$) are defined by replacing $x$, $\ep_1$ in $R_{s_i}^x(\al_i)$, 
with $q$, $\ep_2$, respectively. 
\end{dfn}

\begin{prop}
[\cite{N Weyl}]\label{prop H_II sym}
The automorphisms $R^x_{s_i}(\al_i)$ ($i=0,1$),  $R_\pi$ preserve
 the Hamiltonian 
$H^x_\mathrm{II}(\al)$ in the following sense. 
\begin{align*}
&R^x_{s_i}(\al_i)\left(  H^x_\mathrm{II}(\al) \right)=H^x_\mathrm{II}(s_i(\al))+C_{s_i},
\\ 
&R_\pi\left(  H^x_\mathrm{II}(\al) \right)=H^x_\mathrm{II}(\pi(\al))+C_{\pi},
\end{align*}
where 
\begin{equation*}
C_{s_0}=-{\ka \al_0\over f},
\quad C_{s_1}=0,
\quad C_{\pi}=-\ka x.  
\end{equation*}
\end{prop}
By definition, the automorphisms $R_{s_i}^q(\al_i)$ ($i=0,1$) and $R_{\pi}$ 
 act the Hamiltonian $H^q_\mathrm{II}(\al)$ in the same way above.

\begin{dfn}
Let the automorphisms  
$T_{\pi s_1 }$, $S$ on $\mathcal{K}$ be defined by
\begin{align*}
&T_{\pi s_1 }=R^x_{s_1}(-\al_1-\ka)R_\pi 
R^q_{s_1}(\al_1+\ka),
\\
&S=\Ad(D(\al_1)^{-1})R^x_{s_1}(\al_1)R^q_{s_1}(\al_1+\ka). 
\end{align*}
\end{dfn}

\begin{thm}\label{thm II}
The automorphisms $R_{\pi}$,   $T_{\pi s_1 }$ and $S$ act the quantum Lax operators 
$L_\mathrm{II}(\al)$ and $B_\mathrm{II}(\al)$ as follows. 

For the automorphism $R_\pi$,   
\begin{equation*}
R_{\pi}\left(
L_\mathrm{II}(\al),B_\mathrm{II}(\al)
\right)=\left(
L_\mathrm{II}(\pi(\al)), B_\mathrm{II}(\pi(\al))
\right). 
\end{equation*}

For the automorphism $T_{\pi s_1 }$,  
\begin{align}
&l_{T_{\pi s_1 }} T_{\pi s_1 }\left((x-q)
L_\mathrm{II}(\al)
\right)=(q-x)L_\mathrm{II}\left(\pi s_1
(\al)
\right)
,\label{eq II euler}
\\
&T_{\pi s_1 }\left(
B_\mathrm{II}(\al)
\right)=B_\mathrm{II}\left(\pi s_1
(\al)
\right),\nn
\end{align}
 where $l_{T_{\pi s_1}}$ is  some element in $\mathcal{K}$ whose explicit form is given in the proof. 
%\begin{align*}
%&l_{T_{\pi s_1 }}= R_\pi\Ad\left(\exp\left(\left( {2\over 3}x^3+xt\right){1\over \ep_1}\right)\right)R^q_{s_1}(\al_1+\ka)
%\left(\left(\left(
%(x-q)y+\al_1+\ka
%\right)
%p\right)^{-1}
%\left(
%(x-q)p-\al_1-\ka
%\right)
%y\right). 
%\end{align*}

For the automorphism $S$, 
\begin{align*}
&yp\left(R^x_{s_1}(\al_1)R^q_{s_1}(\al_1+\ka)\left((x-q)L_\mathrm{II}(\al)\right)\right)D(\al_1)
\\
&=\left(
(x-q)yp+(\al_1+\ka-\ep_2)y+(\ep_1-\al_1)p
\right)
\left(
D(\al_1)-{(\al_1+\ep_1)(\ep_1+\ep_2)\over (x-q)^2}
\right)L_\mathrm{II}(\al_1+\ep_1+\ep_2,-\al_1-2\ep_1), 
\\
&S\left(
B_\mathrm{II}(\al)
\right)=B_\mathrm{II}(\al_1+\ep_1+\ep_2,-\al_1-2\ep_1)
-D(\al_1)^{-1}{2(\al_1+\ep_1)\ep_1\ep_2\over (x-q)^2}L_\mathrm{II}(\al_1+\ep_1+\ep_2,-\al_1-2\ep_1). 
\end{align*}

\end{thm}

\begin{proof}
For the cases of
the automorphisms $T_{\pi s_1}$  acting $L_\mathrm{II}(\al)$, we show that 
\begin{align}
&\left(
(x-q)p-\al_2-\ka
\right)
y
R_{s_1}^x(-\al_1-\ka)\Ad\left(\exp\left(\left(- {2\over 3}x^3-xt\right){1\over \ep_1}\right)\right)
\left((x-q)
L_\mathrm{{II}}(\al)
\right)\nn
\\
&=-\left(
(x-q)y+\al_2+\ka
\right)
pR^q_{s_1}(-\al_1-\ka)\nn
\\
&\circ 
\left( x\mapsto -x, q\mapsto -q\right)\Ad\left(\exp\left(\left(- {2\over 3}q^3-qt\right){1\over \ep_2}\right)\right)
\left((x-q)
L_\mathrm{{II}}(\pi s_1
(\al)
\right),  
\label{eq II euler2} 
\end{align}
which is the explicit form of \eqref{eq II euler}. 
We omit the proofs of  \eqref{eq II euler2}, since they are similar to that of Theorem \ref{thm VI}.   

Proofs of the other cases
 follow from 
direct computations by using Proposition \ref{prop H_II sym}. 
\end{proof}

Actions involving $R^x_{s_0}(\al_0)$, $R^q_{s_0}(\al_0)$ on the quantum Lax operators can be 
obtained from Theorem \ref{thm II}, because of the definitions of $R^x_{s_0}(\al_0)$, $R^q_{s_0}(\al_0)$.

\def\VV{{\Phi}}
%%%%%%%%%%%
\section{Derivation of the quantum Lax pair from CFT}

In this section, we derive the quantum Lax operators $L_\mathrm{J}$ and $B_\mathrm{J}$ ($\mathrm{J}=\mathrm{I}, \ldots \mathrm{VI}$) from Virasoro conformal field theory. Note that the quantum Lax operators $L_\mathrm{J}$ and $B_\mathrm{J}$ introduced in section 2 are linear 
combinations of $L_\mathrm{J}$ and $B_\mathrm{J}$ in this section, up to gauge transformations, and the parameters $\al_i$ in section 2 
are also linear combinations of $a_i$ in this section (see Remark \ref{re Lax cft}). 

The central charge $c$ and conformal dimension ($L_0$-eigen value) $h$ of the
Virasoro algebra $[L_m, L_n]=(m-n)L_{m+n}+\frac{c}{12}(m^3-m)\delta_{m+n,0}$ are parameterized as \cite{AGT}
\begin{equation}
c=1+6\frac{(\epsilon_1+\epsilon_2)^2}{\epsilon_1\epsilon_2}, \quad
h(\alpha)=\frac{\frac{\alpha}{2}(\epsilon_1+\epsilon_2-\frac{\alpha}{2})}{\epsilon_1\epsilon_2}.
\end{equation}

Following \cite{Nagoya-Sun}, we introduce the $k$-th confluent operator
%\footnote{For the application of these operators with irregular singularity to $4d$
%gauge theories, see \cite{Gaiotto}\cite{AFKMY}\cite{BMT} and references therein.} 
$\VV^{[k]}(z)$, depending on parameters $u_0, \ldots, u_k$ as
\begin{equation}
\VV^{[k]}(z)=\exp\Big\{u_0 \varphi(z)+\frac{u_1}{1!} \varphi'(z)+\cdots+\frac{u_k}{k!} \varphi^{(k)}(z)\Big\},
\end{equation}
where $\varphi(z)$ is a free boson such that $\varphi(z)\varphi(w)=\log(z-w)+{\rm regular}$.
$\VV^{[0]}$ corresponds to the usual primary field.
The OPEs of $\VV^{[k]}$ with $J(z)=\varphi'(z)$ and $T(z)=\frac{1}{2}J(z)^2+\rho J'(z)$ are 
\begin{align*}
&J(z) \VV^{[k]}(w)=\sum_{n}J^{(w)}_n (z-w)^{-n-1}\VV^{[k]}(w)=\{\frac{u_k}{(z-w)^{k+1}}+\frac{u_{k-1}}{(z-w)^{k}}+\cdots\}\VV^{[k]}(w),\\
&J(z) \VV^{[k]}(\infty)=\sum_{n}J^{(\infty)}_n z^{n-1}\VV^{[k]}(\infty)=\{u_k{z^{k-1}}+u_{k-1}{z^{k-2}}+\cdots\}\VV^{[k]}(\infty),\\
&T(z) \VV^{[k]}(w)=\sum_{n}L^{(w)}_n (z-w)^{-n-2}\VV^{[k]}(w)=\{\frac{u_k^2}{2(z-w)^{2k+2}}+\frac{u_k u_{k-1}}{(z-w)^{2k+1}}+\cdots\}\VV^{[k]}(w),\\
&T(z) \VV^{[k]}(\infty)=\sum_{n}L^{(\infty)}_n z^{n-2}\VV^{[k]}(\infty)=\{\frac{u_k^2}{2}{z^{2k-2}}+u_k u_{k-1}{z^{2k-3}}+\cdots\}\VV^{[k]}(\infty).
\end{align*}
More explicitly, in case of $k=3$ for instance, we have
\begin{align}\label{eq:irrk3eg}
T(z) &\VV^{[k]}_{u_0,\cdots,u_3}(\infty)=\{\frac{u_3^2}{2}{z^{4}}+u_3 u_{2}{z^3}+(\frac{u_2^2}{2}+u_3u_1)z^2+(u_2u_1+u_3u_0+2\rho u_3)z\nonumber \\ 
&+(\frac{u_1^2}{2}+u_2u_0+u_3\frac{\partial}{\partial u_1}+\rho u_2)
+(u_1u_0+u_2\frac{\partial}{\partial u_1}+u_3\frac{\partial}{\partial u_2})z^{-1}+\cdots\}\VV^{[k]}(\infty).
\end{align}

\subsection{$\mathrm{P}_\mathrm{{VI}}$ case}

Let $\Psi^{CFT}_\mathrm{{VI}}(q,x,t)$ be a correlation function on ${\mathbb P}^1$ defined as
\begin{equation}
\Psi^{CFT}_\mathrm{{VI}}=\langle {\mathcal O}_\mathrm{{VI}}\rangle,\quad 
{\mathcal O}_\mathrm{{VI}}=\VV_{h_0}(0)\VV_{h_1}(1)\VV_{h_t}(t)\VV_{h_{\infty}}(\infty)\VV_{h_q}(q)\VV_{h_x}(x),
\end{equation}
where $\VV_{h_i}$ is the primary field of dimension $h_i=h(a_i)$, $(i=0,1,t,\infty, q,x)$.
We put\footnote{We apply these specializations also for $\mathrm{J}=\mathrm{II}, \cdots, \mathrm{V}$ cases below.} $a_q=-\epsilon_1$ and $a_x=-\epsilon_2$, then we have the null field constraints
\begin{equation}
L_{-2}^{(q)}\VV_{h_q}(q)=-\frac{\epsilon_2}{\epsilon_1}\frac{\partial^2}{\partial q^2}\VV_{h_q}(q), \quad
L_{-2}^{(x)}\VV_{h_x}(x)=-\frac{\epsilon_1}{\epsilon_2}\frac{\partial^2}{\partial x^2}\VV_{h_x}(x).
\end{equation}

From the residue theorem
$\sum_{\rm poles} {\rm Res} \Big(\xi(z) \langle T(z){\mathcal O} \rangle dz \Big)=0$
for the vector field
$\xi_{L_\mathrm{{VI}}}(z)\frac{\partial}{\partial z}=\frac{z(z-1)(z-t)}{(z-q)(z-x)}\frac{\partial}{\partial z}$,
we obtain a linear relation between
$\{L_0^{(i)}\}_{i=0,1,t,\infty,q,x}$, $L_{-1}^{(q)}, L_{-2}^{(q)}, L_{-1}^{(x)}$ and $L_{-2}^{(x)}$ which
gives a differential equation for $\Psi^{CFT}_\mathrm{{VI}}$ of second order in $q$ and $x$.
Under the gauge transformation $\Psi^{CFT}_{J}=g_{J} \Psi_{J}$ where
$g_J=(x-q)^{-\frac{1}{2}}f_{J}(x)^{\frac{1}{2 \epsilon_1}}f_{J}(q)^{\frac{1}{2\epsilon_2}}$ with $f_\mathrm{{VI}}(z)=z^{a_0}(z-1)^{a_1}(z-t)^{a_t}$, we obtain the desired equation $L_\mathrm{{VI}}\Psi_\mathrm{{VI}}=0$. 
Similarly, taking the vector field as $\xi_{B_\mathrm{{VI}}}(z)=\frac{z(z-1)}{z-q}$, we have the deformation equation
$B_\mathrm{{VI}}\Psi_\mathrm{{VI}}=0$. The final results are as follows
\begin{align*}
L_\mathrm{{VI}}=&-(x-1) x  (x-t) \Big\{\frac{a_t-{\epsilon_2}}{x-t}+\frac{a_0-{\epsilon_2}}{x}+\frac{a_1-{\epsilon_2}}{x-1}+\frac{{\epsilon_2}-{\epsilon_1}}{x-q}\Big\}{\epsilon_1}\partial_x\\
&+(q-1) q  (q-t)  \Big\{\frac{a_t-{\epsilon_1}}{q-t}+\frac{a_0-{\epsilon_1}}{q}+\frac{a_1-{\epsilon_1}}{q-1}+\frac{{\epsilon_1}-{\epsilon_2}}{q-x}\Big\}{\epsilon_2}\partial_q\\
&+{C}(q-x) -(x-1) x (x-t) {\epsilon_1}^2{\partial_x}^2+(q-1) q (q-t) {\epsilon_2}^2 {\partial_q}^2,\\
B_\mathrm{{VI}}=&(q-1) q \Big\{-\frac{a_t}{q-t}+\frac{{\epsilon_1}-a_0}{q}+\frac{{\epsilon_1}-a_1}{q-1}+\frac{{\epsilon_2}}{q-x}\Big\}{\epsilon_2}\partial_q\\
&-\Big\{\frac{\left(a_0 t+a_1 t-a_0\right) a_t}{2 (q-t)}+{C} \Big\}
-\frac{(t-1) t}{q-t}{\epsilon_1} {\epsilon_2} \partial_t-\frac{(x-1) x}{q-x}{\epsilon_1} {\epsilon_2}\partial_x-(q-1) q {\epsilon_2}^2 {\partial_q}^2.
\end{align*}
where $C=(-a_t-a_{\infty }-a_0-a_1+3 {\epsilon_1}+3 {\epsilon_2})(-a_t+a_{\infty }-a_0-a_1+{\epsilon_1}+{\epsilon_2})/4$. 
We note that the deformation equation $B_\mathrm{VI}\Psi_\mathrm{VI}=0$ 
is equivalent to the BPZ equation (5.17)  
in \cite{BPZ} associated to the field $\Phi_{h_q}(q)$ and $L_\mathrm{VI}\Psi_\mathrm{VI}=0$ is a 
linear combination of the BPZ equations %(5.17) in \cite{BPZ} 
associated to the fields $\Phi_{h_q}(q)$ and $\Phi_{h_x}(x)$. 

\begin{re}\label{re Lax cft}
 Denote by $\widehat{L}_\mathrm{VI}$ and $\widehat{B}_\mathrm{VI}$ 
the quantum Lax operators defined by \eqref{eq Lax LVI} and \eqref{eq Lax BVI}, respectively. 
The quantum Lax operators $\widehat{L}_\mathrm{VI}$ and $\widehat{B}_\mathrm{VI}$ are expressed in  terms of $L_\mathrm{VI}$ and $B_\mathrm{VI}$ 
from Virasoro conformal field theory as follows:
\begin{align*}
\widehat{L}_\mathrm{VI}=&-L_\mathrm{VI},
\\
\widehat{B}_\mathrm{VI}=&-\ep_2 L_\mathrm{VI}+(\ep_1-\ep_2)(q-t)B_\mathrm{VI}+b
,
\end{align*}
where 
\begin{equation*}
y=\ep_1\partial_x,\quad p=\ep_2\partial_q,\quad d=\partial_t,\quad 
\begin{pmatrix}
\al_0\\ \al_1\\ \al_3 \\ \al_4
\end{pmatrix}=
\begin{pmatrix}
-a_t\\ -a_\infty \\ -a_1 \\ -a_0
\end{pmatrix}
+(\ep_1+\ep_2)
\begin{pmatrix}
1\\1\\1\\1
\end{pmatrix},
\end{equation*}
and $b=(\ep_1-\ep_2)\{\left(a_0 t+a_1 t-a_0\right) a_t/2-{C}t \}$ can be removed by some 
gauge transformation. 
\end{re}

\subsection{$\mathrm{P}_\mathrm{{V}}$ case}

Operators: ${\mathcal O}_\mathrm{{V}}=\VV_{h_0}(0)\VV_{h_1}(1)\VV_\mathrm{{V}}(\infty)\VV_{h_q}(q)\VV_{h_x}(x)$,
where $\VV_\mathrm{{V}}\in \{\VV^{[1]}\}$ such as
\begin{equation}
T(z)\VV_\mathrm{{V}}(\infty)=\{\frac{-t^2}{4 \epsilon_1 \epsilon_2}
+\frac{t (\epsilon_1+\epsilon_2+2 a_2-a_0-a_1)}{2 \epsilon_1 \epsilon_2}z^{-1}+t \frac{\partial}{\partial t}z^{-2}+\cdots\}\VV_\mathrm{{V}}(\infty).
\end{equation}
Vector fields:
$\xi_{L_\mathrm{{V}}}(z)=\frac{z(z-1)}{(z-q)(z-x)}$, 
$\xi_{B_\mathrm{{V}}}(z)=\frac{z(z-1)}{z-q}$. 
Gauge factor:
$f_\mathrm{{V}}(z)=z^{a_0}(z-1)^{a_1}e^{t z}$.
\begin{align*}
L_V=&(x-1) x \Big\{\frac{a_0-{\epsilon_2}}{x}+\frac{a_1-{\epsilon_2}}{x-1}+\frac{{\epsilon_2}-{\epsilon_1}}{x-q}+t\Big\}{\epsilon_1} \partial_x -
(q-1) q  \Big\{\frac{a_0-{\epsilon_1}}{q}+\frac{a_1-{\epsilon_1}}{q-1}+\frac{{\epsilon_1}-{\epsilon_2}}{q-x}+t\Big\}{\epsilon_2} \partial_q\\
&-t a_2  (q-x)  +(x-1) x {\epsilon_1}^2 {\partial_x}^2-(q-1) q {\epsilon_2}^2 {\partial_q}^2,\\
B_V=&(q-1) q \Big\{\frac{{\epsilon_1}-a_0}{q}+\frac{{\epsilon_1}-a_1}{q-1}+\frac{{\epsilon_2}}{q-x}-t\Big\}{\epsilon_2} \partial_q -
t{\epsilon_1} {\epsilon_2} \partial_t-\frac{(x-1) x}{q-x}{\epsilon_1} {\epsilon_2} \partial_x-(q-1) q {\epsilon_2}^2 {\partial_q}^2\\
&+\Big\{\frac{1}{2} t\left(-a_1-2 a_2 q  +2 a_2 \right)-\frac{1}{4} \left(-a_0-a_1+{\epsilon_1}+{\epsilon_2}\right) \left(-a_0-a_1+3 {\epsilon_1}+3 {\epsilon_2}\right)\Big\}.
\end{align*}

\subsection{$\mathrm{P}_\mathrm{{IV}}$ case}

Operators: ${\mathcal O}_\mathrm{{IV}}=\VV_{h_0}(0)\VV_\mathrm{{VI}}(\infty)\VV_{h_q}(q)\VV_{h_x}(x)$, 
where $\VV_\mathrm{{VI}}\in \{\VV^{[2]}\}$ such as
\begin{equation}
T(z) \VV_\mathrm{{VI}}(\infty)=\{
\frac{-1}{16 \epsilon_1 \epsilon_2}z^2+
\frac{-t}{4 \epsilon_1 \epsilon_2}z+
\frac{-t^2+2 a_1-a_0}{4 \epsilon_1 \epsilon_2}+
\frac{1}{2}\partial_t z^{-1}+\cdots\}\VV_\mathrm{{VI}}(\infty) ,
\end{equation}
Vector fields:
$\xi_{L_\mathrm{{VI}}}(z)=\frac{z}{(z-q)(z-x)}$,
$\xi_{B_\mathrm{{VI}}}(z)=\frac{z}{z-q}$.
Gauge factor:
$f_\mathrm{{VI}}(z)=z^{a_0}e^{-t z-z^2/4}$.
\begin{align*}
L_\mathrm{{IV}}
&=x \Big\{\frac{a_0-{\epsilon_2}}{x}+\frac{{\epsilon_2}-{\epsilon_1}}{x-q}-t-\frac{x}{2}\Big\}{\epsilon_1} \partial_x 
-q  \Big\{\frac{a_0-{\epsilon_1}}{q}+\frac{{\epsilon_1}-{\epsilon_2}}{q-x}-\frac{q}{2}-t\Big\}{\epsilon_2} \partial_q\\
&+x {\epsilon_1}^2 {\partial_x}^2-q {\epsilon_2}^2 {\partial_q}^2+\frac{1}{2} a_1  (q-x),\\
B_\mathrm{{IV}}&=q  \Big\{\frac{{\epsilon_1}-a_0}{q}+\frac{{\epsilon_2}}{q-x}+\frac{q}{2}+t\Big\}{\epsilon_2} \partial_q
+\frac{1}{2} \Big\{t (a_0-{\epsilon_1}-{\epsilon_2})+a_1 q  \Big\}
-\frac{1}{2} {\epsilon_1} {\epsilon_2} \partial_t-\frac{x}{q-x}{\epsilon_1} {\epsilon_2} \partial_x-q {\epsilon_2}^2 {\partial_q}^2.
\end{align*}

\subsection{$\mathrm{P}_\mathrm{{III}}$ case}
Operators:
${\mathcal O}_\mathrm{{III}}=\VV_\mathrm{{III}}(0)\VV_{\mathrm{III}'}(\infty)\VV_{h_q}(q)\VV_{h_x}(x)$, 
where $\VV_\mathrm{{III}}(0), \VV_{\mathrm{III}'}(\infty)\in \{\VV^{[1]}\}$ such as
\begin{align*}
&T(z) \VV_\mathrm{{III}}(0)=\{
\frac{-t^2}{4 \epsilon_1 \epsilon_2}z^{-4}
+\frac{t(2\epsilon_1+2\epsilon_2-a_0)}{2 \epsilon_1 \epsilon_2}z^{-3}
+t\partial_t z^{-2}+\cdots\}\VV_\mathrm{{III}}(0),\\
&T(z)\VV_{\mathrm{III}'}(\infty)=\{
\frac{-1}{4 \epsilon_1 \epsilon_2}
+\frac{-\epsilon_1-\epsilon_2-2a_1+a_0}{2 \epsilon_1 \epsilon_2}z^{-1}
+\cdots\}\VV_{\mathrm{III}'}(\infty).
\end{align*}
Vector fields: $\xi_{L_\mathrm{{III}}}(z)=\frac{z}{(z-q)(z-x)}$,
$\xi_{B_\mathrm{{III}}}(z)=\frac{z}{z-q}$.
Gauge factor: $ f_\mathrm{{III}}(z)=z^{\alpha_0}e^{t/z+z}$.
\begin{align*}
L_\mathrm{{III}}&=-q^2 \Big\{\frac{a_0-2 {\epsilon_1}}{q}+\frac{t}{q^2}+\frac{{\epsilon_1}-{\epsilon_2}}{q-x}-1\Big\}{\epsilon_2} \partial_q+
x^2  \Big\{\frac{a_0-2 {\epsilon_2}}{x}+\frac{{\epsilon_2}-{\epsilon_1}}{x-q}+\frac{t}{x^2}-1\Big\}{\epsilon_1} \partial_x\\
&-q^2 {\epsilon_2}^2 {\partial_q}^2+x^2 {\epsilon_1}^2 {\partial_x}^2+a_1  (q-x),\\
B_\mathrm{{III}}&=-q^2 \Big\{\frac{a_0-{\epsilon_1}}{q}+\frac{t}{q^2}-\frac{{\epsilon_2}}{q-x}-1\Big\}{\epsilon_2} \partial_q +
\Big\{\frac{1}{4} a_0 (-a_0+2 {\epsilon_1}+2 {\epsilon_2})+a_1 q  +\frac{t}{2}\Big\}\\
&-q^2 {\epsilon_2}^2 {\partial_q}^2-\frac{q x}{q-x}{\epsilon_1} {\epsilon_2}\partial_x-t {\epsilon_1} {\epsilon_2} \partial_t.
\end{align*}

\subsection{$\mathrm{P_{\mathrm{II}}}$ case}

Operators: ${\mathcal O}_\mathrm{II}=\VV_\mathrm{II}(\infty)\VV_{h_q}(q)\VV_{h_x}(x)$, 
where $\VV_\mathrm{II}\in\{\VV^{[3]}\}$ such that\footnote{We have set the additional parameter $u_2$=0 in the corresponding equation (\ref{eq:irrk3eg}).}
\begin{equation}
T(z) \VV_\mathrm{II}(\infty)=\{
\frac{-1}{\epsilon_1 \epsilon_2}z^4
+\frac{-t}{16 \epsilon_1 \epsilon_2}z^2
+\frac{-2a-\epsilon_1+\epsilon_2}{\epsilon_1 \epsilon_2}z
+2\partial_t
+\frac{2a+\epsilon_1-\epsilon_2}{\epsilon_1 \epsilon_2}t z^{-1}
+\cdots\} \VV_\mathrm{II}(\infty).
\end{equation}
Vector fields:
$\xi_{L_\mathrm{II}}(z)=\frac{1}{2(z-q)(z-x)}$, 
$\xi_{B_\mathrm{II}}(z)=\frac{1}{z-q}$.
Gauge factor:
$f_\mathrm{II}(z)=e^{-t z-\frac{2}{3}z^3}$.
\begin{align*}
&L_\mathrm{II}=2 (a+{\epsilon_1}) (q-x)
+\Big\{2 q^2+t+\frac{{\epsilon_2}-{\epsilon_1}}{q-x}\Big\}{\epsilon_2} \partial_q 
-\Big\{2 x^2+t+\frac{{\epsilon_1}-{\epsilon_2}}{x-q}\Big\}{\epsilon_1} \partial_x
+{\epsilon_1}^2 {\partial_x}^2-{\epsilon_2}^2 {\partial_q}^2,\\
&B_\mathrm{II}=2(a+{\epsilon_1})q+\Big\{2 q^2+t+\frac{{\epsilon_2}}{q-x}\Big\} {\epsilon_2}\partial_q
-2 {\epsilon_1} {\epsilon_2} \partial_t-\frac{{\epsilon_2}}{q-x}{\epsilon_1}\partial_x-{\epsilon_2}^2 {\partial_q}^2.
\end{align*}

\subsection{$\mathrm{P}_\mathrm{I}$ case}

Operators: ${\mathcal O}_\mathrm{I}=\VV_\mathrm{I}(\infty)\VV_{h_q}(q)\VV_{h_x}(x)$, 
where $\VV_\mathrm{I}$ is a degenerate case of $\VV^{[3]}$ such that
\begin{equation}
T(z) \VV_\mathrm{I}(\infty)=\{
\frac{-4}{\epsilon_1 \epsilon_2}z^3
+\frac{-2t}{\epsilon_1 \epsilon_2}z
+2\partial_t+\cdots\} \VV_\mathrm{I}(\infty).
\end{equation}
Vector fields:
$\xi_{L_\mathrm{I}}(z)=\frac{1}{2(z-q)(z-x)}$, 
$\xi_{B_\mathrm{I}}(z)=\frac{1}{z-q}$.
Gauge factor:
$f_\mathrm{I}(z)=1$.
\begin{align*}
&L_\mathrm{I}=(4 q^3 +2 q t-4 x^3-2 x t)-\frac{{\epsilon_1}-{\epsilon_2}}{q-x}{\epsilon_1} \partial_x 
+\frac{{\epsilon_2}-{\epsilon_1}}{x-q}{\epsilon_2} \partial_q+{\epsilon_1}^2 {\partial_x}^2-{\epsilon_2}^2 {\partial_q}^2,\\
&B_\mathrm{I}=(4 q^3+2 q t)+\frac{1}{q-x}{\epsilon_2}^2\partial_q
-\frac{1}{q-x}{\epsilon_1} {\epsilon_2} \partial_x-{\epsilon_2}^2 {\partial_q}^2-2{\epsilon_1} {\epsilon_2} \partial_t.
\end{align*}

\subsection{$\mathrm{P}_\mathrm{{III}}$ ($D_7$) case}
Operators: ${\mathcal O}_\mathrm{{III}}^{(D_7)}=\VV_\mathrm{{III}}^{(D_7)}(0)\VV_{\mathrm{III}'}^{(D_7)}(\infty)\VV_{h_q}(q)\VV_{h_x}(x)$, 
\begin{align*}
&T(z) \VV_\mathrm{{III}}^{(D_7)}(0)=\{
\frac{-t^2}{4 \epsilon_1 \epsilon_2}z^{-4}
+\frac{t(2\epsilon_1+\epsilon_2-a_0)}{2 \epsilon_1 \epsilon_2}z^{-3}
+t\partial_t z^{-2}+\cdots\}\VV_\mathrm{{III}}^{(D_7)}(0),\\
&T(z)\VV_{\mathrm{III}'}^{(D_7)}(\infty)=\{
\frac{1}{\ep_1\epsilon_2}z^{-1}+\cdots\}\VV_{\mathrm{III}'}^{(D_7)}(\infty).
\end{align*}
Vector fields: $\xi_{L_\mathrm{{III}}^{(D_7)}}(z)=\frac{z^2}{(z-q)(z-x)}$,
$\xi_{B_\mathrm{{III}}^{(D_7)}}(z)=\frac{z}{z-q}$.
Gauge factor: $ f_\mathrm{{III}}^{(D_7)}(z)=e^{-t/(2 z)}z^{a_0}$.

\begin{align*}
&L_\mathrm{{III}}^{(D_7)}= 
\Big\{q (2 {\epsilon_1}-a_0)-\frac{q^2 ({\epsilon_1}-{\epsilon_2})}{q-x}-t\Big\}{\epsilon_2} {\partial_q}
+\Big\{x(a_0-2 {\epsilon_2})+\frac{x^2 ({\epsilon_1}-{\epsilon_2})}{q-x}+t\Big\}{\epsilon_1} {\partial_x}\\
&-q^2 {\epsilon_2}^2 {\partial_q}^2-\frac{(q-x)  (2 q x +t {\epsilon_2})}{2 q x}+x^2 {\epsilon_1}^2 {\partial_x}^2,\\
&B_\mathrm{{III}}^{(D_7)}=
\Big\{-a_0 q+\frac{q^2 {\epsilon_2}}{q-x}+q {\epsilon_1}-t\Big\}{\epsilon_2} {\partial_q}
+\Big\{\frac{1}{4} a_0 (-a_0+2 {\epsilon_1}+2{\epsilon_2})+\frac{t {\epsilon_2}}{2 q}-q \Big\}\\
&-q^2 {\epsilon_2}^2 {\partial_q}^2-t {\epsilon_1} {\epsilon_2}{\partial_t}-\frac{q x}{q-x}{\epsilon_1} {\epsilon_2} {\partial_x}.
\end{align*}

\subsection{$\mathrm{P}_\mathrm{{III}}$ ($D_8$) case}

Operators: ${\mathcal O}_\mathrm{{III}}^{(D_8)}=\VV_\mathrm{{III}}^{(D_8)}(0)\VV_{\mathrm{III}'}^{(D_8)}(\infty)\VV_{h_q}(q)\VV_{h_x}(x)$, 
\begin{align*}
&T(z) \VV_\mathrm{{III}}^{(D_8)}(0)=\{
\frac{t}{\epsilon_1 \epsilon_2}z^{-3}
+t\partial_t z^{-2}+\cdots\}\VV_\mathrm{{III}}^{(D_8)}(0),\\
&T(z)\VV_{\mathrm{III}'}^{(D_8)}(\infty)=\{
\frac{1}{\ep_1\epsilon_2}z^{-1}+\cdots\}\VV_{\mathrm{III}'}^{(D_8)}(\infty).
\end{align*}
Vector fields: $\xi_{L_\mathrm{{III}}^{(D_8)}}(z)=\frac{z^2}{(z-q)(z-x)}$,
$\xi_{B_\mathrm{{III}}^{(D_8)}}(z)=\frac{z}{z-q}$.
Gauge factor: $ f_\mathrm{{III}}^{(D_8)}(z)=z^{\epsilon_1+\ep_2}$.

\begin{align*} 
&L_\mathrm{{III}}^{(D_8)}=    
-\frac{qx({\epsilon_1}-{\epsilon_2})}{q-x}{\epsilon_2}{\partial_q}-q^2 {\epsilon_2}^2 {\partial_q}^2+\frac{(q-x)  (t-q x )}{q x}
+x^2{\epsilon_1}^2 {\partial_x}^2+\frac{qx ({\epsilon_1}-{\epsilon_2})}{q-x} {\epsilon_1} {\partial_x},\\
&B_\mathrm{{III}}^{(D_8)}=-q^2 {\epsilon_2}^2 {\partial_q}^2+\left\{{(\ep_1+\ep_2)^2\over 4}-q-\frac{t}{q}\right\}-t {\epsilon_1} {\epsilon_2}
   {\partial_t}-\frac{q x }{q-x}{\epsilon_1} {\epsilon_2} {\partial_x}
   +\frac{q x }{q-x}\ep_2^2\partial_q. 
   %+q  \Big\{\frac{q {\epsilon_2}}{q-x}-{\epsilon_1}\Big\}{\epsilon_2}{\partial_q}.\\
\end{align*}

\begin{re}
It is known that the classical limit of the Knizhnik-Zamolodchikov equations are the Schlesinger equations \cite{R}, \cite{Har}. 
Similarly, all the above operators $L_J, B_J$ from Virasoro conformal field theory give the Lax pair for the classical Painlev\'e equations $\mathrm{P}_J$ (see Appendix A) under the limit
$\epsilon_2\rightarrow 0$  with
$\epsilon_2 \partial_q \rightarrow p$, up to a gauge factor independent of $z$. See \cite{NS}, \cite{Teschner2} for the more detail. 
\end{re}

\begin{re} In a similar way, one can derive the Lax pair for quantum Garnier system of $N$-variables, by inserting $N$-primary
fields $\VV_{-\epsilon_1}(q_i)$ ($i=1,\ldots, N$).
\end{re}

\begin{re}
The confluent/degeneration scheme of the Painlev\'e equation is summarized by the following diagram
\begin{equation}
\begin{array}{ccccccccc}
\mathrm{P}_\mathrm{{VI}}(1,1,1,1) &\rightarrow &\mathrm{P}_\mathrm{{V}}(2,1,1) &\rightarrow &\mathrm{P}_\mathrm{{III}}(2,2) &\rightarrow &\mathrm{P}_\mathrm{{III}}^{D_7}(2,\frac{3}{2}) &\rightarrow &\mathrm{P}_\mathrm{{III}}^{D_8}(\frac{3}{2},\frac{3}{2})\\
                &            &             &\searrow    &             &\searrow    &                       \\
                &            &             &            &\mathrm{P}_\mathrm{{IV}}(3,1)  &\rightarrow &\mathrm{\mathrm{P}_\mathrm{II}}(4)              &\rightarrow &\mathrm{P}_\mathrm{I}(\frac{7}{2})\\
\end{array} 
\end{equation}
where the numbers $(i_1, i_2, \cdots)$ represent the 'Poincar\'e rank $+1$' of the singularities.
The cases $\mathrm{P}_\mathrm{{III}}^{D_7}(2,\frac{3}{2})$, $\mathrm{P}_\mathrm{{III}}^{D_8}(\frac{3}{2},\frac{3}{2})$ are degenerate case of $\mathrm{P}_\mathrm{{III}}$ and studied systematically in \cite{OKSO}.
In view of the $4d$ ${\mathcal N}=2$ gauge theory, the series $(1,1,1,1)\rightarrow (2,1,1)\rightarrow (2,2) \rightarrow (2,\frac{3}{2}) \rightarrow (\frac{3}{2},\frac{3}{2})$
correspond to the $SU(2)$ gauge theories with $N_f=4,3,2,1,0$, and the series $(3,1) \rightarrow (4) \rightarrow (\frac{7}{2})$ corresponds to the AD theories \cite{cubic}.
\end{re}

\appendix
\section{Classical cases}

\subsection{Data for the classical Painlev\'e equations}\cite{O}, \cite{OKSO}
\begin{align*}
\mathrm{P}_\mathrm{I}:\ 
&H_\mathrm{I}=\frac{p^2}{2}-2 q^3-tq,\\
&L_\mathrm{I} =\Big\{-4 x^3-2 t x-2 H_\mathrm{I}+\frac{p}{x-q}\Big\}-\frac{1}{x-q}{\partial_x}+{\partial_x}^2,\\
&B_\mathrm{I} ={\partial_t}-\frac{1}{2 (x-q)}{\partial_x}+\frac{p  }{2 (x-q)},\\
%\end{align*}
\mathrm{P_\mathrm{II}}:\ 
%\begin{align*}
&H_{\mathrm{II}}=\frac{p^2}{2}-\Big\{q^2+\frac{t}{2}\Big\} p-a_1 q,\\
&L_\mathrm{II} =\Big\{\frac{p}{x-q}-2 H_{\mathrm{II}}-2 a_1 x\Big\}-\Big\{2x^2+t+\frac{1}{x-q}\Big\}{\partial_x}+{\partial_x}^2,\\
&B_\mathrm{II} ={\partial_t}-\frac{1}{2(x-q)}{\partial_x}+\frac{p}{2 (x-q)},\\
&s_1 = \{a_1\mapsto -a_1, q\mapsto q+a_1/p\},\\
&\pi = \{a_1\mapsto 1-a_1, q\mapsto-q, p\mapsto -p+2q^2+t\}.\\
%\end{align*}
\mathrm{P}_\mathrm{{III}}:\ 
%\begin{align*}
&H_\mathrm{{III}}=\frac{1}{t}\Big\{p^2 q^2-(q^2+{a_1}q-t)p-{a_0} q\Big\},\\
&L_\mathrm{{III}} =\Big\{-\frac{{a_0}}{x}+\frac{p q}{x (x-q)}-\frac{t H_\mathrm{{III}}}{x^2}\Big\} 
+\Big\{\frac{1-{a_1}}{x}-\frac{1}{x-q}+\frac{t}{x^2}-1\Big\} {\partial_x}+{\partial_x}^2,\\
&B_\mathrm{{III}} ={\partial_t}-\frac{x q}{t (x-q)} {\partial_x}+\frac{pq^2}{t (x-q)},\\ 
&s_0=\{a_0\mapsto-a_0, a_1\mapsto a_1+2 a_0, q\mapsto q+a_0/p\},\\
&s_1=\{a_0\mapsto 1+a_0+a_1, a_1\mapsto-2-a_1, p\mapsto p-(a_1+1)/q+t/q^2, t\mapsto-t\},\\
&s_2=\{a_1\mapsto-2 a_0-a_1, q\mapsto q-(a_0+a_1)/(p-1)\}.\\
%\end{align*}
\mathrm{P}_\mathrm{{III}}^{(D_7)}:\ 
%\begin{align*}
&H_\mathrm{{III}}^{(D_7)}=\frac{1}{t}(p^2 q^2+q+p t+a_1 p q),\\
&L_\mathrm{{III}} =\Big\{\frac{1-p}{x}+\frac{p}{x-q}-\frac{t H_\mathrm{{III}}^{(D_7)}}{x^2}\Big\} 
+\Big\{\frac{a_1+1}{x}-\frac{1}{x-q}+\frac{t}{x^2}\} {\partial_x}+{\partial_x}^2,\\
&B_\mathrm{{III}} ={\partial_t}-\frac{x q}{t (x-q)} {\partial_x}+\frac{pq^2}{t (x-q)},\\ 
&s_0=\{a_1\mapsto 2-a_1, p\mapsto p-(1-a_1)/q+t/q^2, t\mapsto -t\},\\
&s_1=\{a_1\mapsto -a_1, p\mapsto -p, q\mapsto -q-a_1/p-1/p^2, t\mapsto-t\},\\
&\pi=\{a_1\mapsto 1-a_1, q\mapsto tp, p\mapsto -q/t, t\mapsto -t\}.\\
%\end{align*}
\mathrm{P}_\mathrm{{III}}^{(D_8)}:
%\begin{align*}
&H_\mathrm{{III}}^{(D_8)}=\frac{1}{t}(p^2 q^2+pq+q+\frac{t}{q}),\\
&L_\mathrm{{III}}^{(D_8)} =\Big\{\frac{1-p}{x}+\frac{p}{x-q}-\frac{t H_\mathrm{{III}}^{(D_8)}}{x^2}+\frac{t}{x^3}\Big\} 
+\Big\{\frac{2}{x}-\frac{1}{x-q}\Big\} {\partial_x}+{\partial_x}^2,\\
&B_\mathrm{{III}}^{(D_8)} ={\partial_t}-\frac{x q}{t (x-q)} {\partial_x}+\frac{pq^2}{t (x-q)},\\
&\pi=\{q\mapsto t/q, p\mapsto -q(2qp+1)/(2 t)\}.\\
%\end{align*}
\mathrm{P}_\mathrm{{IV}}:\ 
%\begin{align*}
&H_\mathrm{{IV}}=q p f-{a_1} p-{a_2} q,\quad f=p-q-t,\\
&L_\mathrm{{IV}} =\Big\{-{a_2}-\frac{H_\mathrm{{IV}}}{x}+\frac{p q}{x (x-q)}\Big\}
+\Big\{\frac{1-{a_1}}{x}-t-x-\frac{1}{x-q}\Big\} {\partial_x}+{\partial_x}^2,\\
&B_\mathrm{{IV}} ={\partial_t}-\frac{x}{x-q}{\partial_x}+\frac{p q  }{x-q},\\
&s_0=\{p\mapsto  p+(1-a_1-a_2)/f, q\mapsto  q + (1-a_1-a_2)/f, a_1\mapsto  1-a_2, a_2\mapsto  1-a_1\},\\
&s_1=\{p\mapsto  p-a_1/q, a_1\mapsto  -a_1, a_2\mapsto  a_1+a_2\},\\
&s_2=\{q\mapsto  q + a_2/p, a_1\mapsto  a_1+a_2, a_2\mapsto  -a_2\},\\
&\pi=\{p\mapsto  -f, q\mapsto -p, a_1\mapsto  a_2, a_2\mapsto  1-a_1-a_2\}.\\
%\end{align*}
\mathrm{P}_\mathrm{{V}}:\ 
%\begin{align*}
&H_\mathrm{{V}}=\frac{1}{t}\Big\{(q-1) q (p+t) p+\{{a_1}-({a_1}+{a_3}) q\}p+{a_2} q t\Big\},\\
&L_\mathrm{{V}}=\Big\{\frac{p (q-1) q}{(x-1) x (x-q)}+\frac{{a_2}tx-tH_\mathrm{{V}}}{(x-1) x}\Big\}  
+\Big\{\frac{1-{a_1}}{x}+t+\frac{1-{a_3}}{x-1}-\frac{1}{x-q}\Big\}{\partial_x}+{\partial_x}^2,\\
&B_\mathrm{{V}}={\partial_t}-\frac{(x-1) x}{t(x-q)} {\partial_x}+\frac{p (q-1) q  }{t (x-q)},\\
&s_0=\{a_1\mapsto 1-a_2-a_3, a_3\mapsto 1-a_1-a_2, q\mapsto q+(1-a_1-a_2-a_3)/(p+t)\},\\
&s_1=\{a_1\mapsto -a_1, a_2\mapsto a_1+a_2, a_3\mapsto a_3, p\mapsto p-a_1 /q\},\\
&s_2=\{a_1\mapsto a_1+a_2, a_3\mapsto a_2+a_3, a_2\mapsto -a_2, q\mapsto q+a_2/p\},\\
&s_3=\{a_3\mapsto -a_3, a_2\mapsto a_2+a_3, p\mapsto p-a_3/(q-1)\},\\
&\pi=\{a_1\mapsto a_2, a_2\mapsto a_3, a_3\mapsto 1-a_1-a_2-a_3, q\mapsto -p/t, p\mapsto (q-1)t\}.\\
%\end{align*}
\mathrm{P}_\mathrm{{VI}}:\ 
%\begin{align*}
&H_\mathrm{{VI}}=\frac{q(q-1)(q-t)}{t(t-1)}\Big\{p^2-(\frac{a_0-1}{q-t}+\frac{a_3}{q-1}+\frac{a_4}{q}) p\Big\}
+\frac{(q-t) a_2 (a_1+a_2)}{t(t-1)},\\
&L_\mathrm{{VI}}=\Big\{\frac{p (q-1) q}{(x-1) x (x-q)}+\frac{a_2 (a_1+a_2)}{(x-1) x}-\frac{t(t-1)H_\mathrm{{VI}}}{(x-1) x (x-t)}\Big\} \\
&+\Big\{\frac{1-a_0}{x-t}+\frac{1-a_3}{x-1}+\frac{1-a_4}{x}-\frac{1}{x-q}\Big\} {\partial_x}+{\partial_x}^2,\\
&B_\mathrm{{VI}}={\partial_t}-\frac{(t-q)(x-1) x}{(t-1) t (q-x)}{\partial_x}+\frac{p (q-1) q (q-t)  }{(t-1) t (x-q)},
\end{align*}
\begin{align*}
&s_0=\{a_0\mapsto -a_0, a_2\mapsto a_0+a_2, p\mapsto p-a_0/(q-t)\}, \\
&s_1=\{a_1\mapsto -a_1, a_2\mapsto a_1+a_2\}, \\
&s_2=\{a_0\mapsto a_0+a_2, a_1\mapsto a_1+a_2, a_2\mapsto -a_2, a_3\mapsto a_2+a_3, a_4\mapsto a_2+a_4, q\mapsto q+a_2/p\}, \\
&s_3=\{a_2\mapsto a_2+a_3, a_3\mapsto -a_3, p\mapsto p-a_3/(q-1)\}, \\
&s_4=\{a_2\mapsto a_2+a_4, a_4\mapsto -a_4, p\mapsto p-a_4/q\}, \\
&\pi_1=\{a_0\mapsto a_1, a_1\mapsto a_0, a_3\mapsto a_4, a_4\mapsto a_3, 
p\mapsto -\frac{(q-t)(p(q-t)+a_2)}{t(t-1)}, q\mapsto \frac{(q-1)t}{q-t}\}, \\
&\pi_2=\{a_0\mapsto a_3, a_1\mapsto a_4, a_3\mapsto a_0, a_4\mapsto a_1, p\mapsto -\frac{q(pq+a_2)}{t}, q\mapsto \frac{t}{q}\}, \\
&\pi_3=\{a_0\mapsto a_4, a_1\mapsto a_3, a_3\mapsto a_1, a_4\mapsto a_0, 
p\mapsto \frac{(q-1)(p(q-1)+a_2)}{t-1}, q\mapsto \frac{q-t}{q-1}\}.  
\end{align*}  
   
\subsection{Symmetry of the classical Lax operator}\cite{Kawakami}, \cite{Takemura}

\begin{prop}
If $L_\mathrm{J} y(x)=0$ then $\ell w(L_\mathrm{J}) {\tilde y}=0$, where
\begin{align*}
&{\tilde y}=(\partial_x)^{2-a_1}y,\ w=\pi s_1\pi,\ \ell=\partial_x+\frac{1-a_1}{x-q}+\frac{1}{x-w(q)},\quad ({\rm for}\ \mathrm{J=II})\\
&{\tilde y}=(\partial_x)^{2-a_1}e^{1\over \partial_x}y,\ w=\pi s_1\pi s_1,\ \ell=\partial_x^2-\left(  {1\over x-q}+{a_1-1\over s_1(q)} \right)\partial_x
-\frac{p}{x-q}+\frac{p+1}{x+s_1(q)},\quad ({\rm for}\ \mathrm{J=III}^{D_7})\\
&{\tilde y}=(\partial_x)^{2-a_0}y,\ w=s_1 s_2 s_1 s_0,\ \ell=\partial_x+\frac{2}{x}+\frac{1-a_0}{x-q}+\frac{1}{x-w(q)},\quad ({\rm for}\ \mathrm{J=III})\\
&{\tilde y}=(\partial_x)^{2-a_2}y,\ w=s_1 s_0 s_1 ,\ \ell=\partial_x+\frac{1}{x}+\frac{1-a_2}{x-q}+\frac{1}{x-w(q)},\quad ({\rm for}\ \mathrm{J=IV})\\
&{\tilde y}=(\partial_x)^{2-a_2}y,\ w=s_3 s_0 s_1 s_0 s_3,\ \ell=\partial_x+\frac{1}{x}+\frac{1}{x-1}+\frac{1-a_2}{x-q}+\frac{1}{x-w(q)},\quad ({\rm for}\ \mathrm{J=V})\\
&{\tilde y}=(\partial_x)^{2-a_2}y,\ w=s_4 s_3 s_1 s_0 s_2 s_4 s_3 s_1 s_0,\ \ell=\partial_x+\frac{1}{x}+\frac{1}{x-1}+\frac{1}{x-t}+\frac{1-a_2}{x-q}+\frac{1}{x-w(q)}.\quad ({\rm for}\ \mathrm{J=VI})\\
\end{align*}
\end{prop}

\begin{prop}
If $L_J y(x)=0$ then $w(L_J) {\tilde y}=0$, where
\begin{align*}
&{\tilde y}=\{a_1 y+(x-q-a_1/p)y_x\}/{(x-q)},\ w=s_1 \pi s_1\pi,\quad ({\rm for}\ \mathrm{J=II})\\
&{\tilde y}=\{x y_x-qpy\}/{(x-q)},\ w= s_1\pi,\quad ({\rm for}\ \mathrm{J=III}^{D_7})\\
&{\tilde y}=\{a_0 y+(x-q-a_0/p)y_x\}/{(x-q)},\ w=s_1 s_2 s_1,\quad ({\rm for}\ \mathrm{J=III})\\
&{\tilde y}=\{a_2 y+(x-q-a_2/p)y_x\}/{(x-q)},\ w=s_1 s_0 s_1 s_2,\quad ({\rm for}\ \mathrm{J=IV})\\
&{\tilde y}=\{a_2 y+(x-q-a_2/p)y_x\}/{(x-q)},\ w=s_3 s_0 s_1 s_0 s_3 s_2,\quad ({\rm for}\ \mathrm{J=V})\\
&{\tilde y}=\{a_2 y+(x-q-a_2/p)y_x\}/{(x-q)},\ w=s_4 s_3 s_1 s_0 s_2 s_4 s_3 s_1 s_0 s_2.\quad ({\rm for}\ \mathrm{J=VI})\\
\end{align*}
\end{prop}

\bigskip

{\bf Acknowledgements}
This work was partially supported by JSPS Grant-in-Aid for Scientific Research 21340036 and Grant-in-Aid for JSPS Fellows 22-2255.

%%%%%%%%%%%%%%%%%%%%%%%%%%%%%%%%%%%%%%%%%%%%


\begin{thebibliography}{[FJKLM]}
%\bibitem{AGT}

%\bibitem{ATY}
%H.~Awata, A.~Tsuchiya and Y.~Yamada,
%Integral formulas for the WZNW Correlation Fucntions,
%{\em Nucl.~Phys.} {\bf B365} (1991) 680--696

\bibitem{AGT}
L.~F.~Alday, D.~Gaiotto and Y.~Tachikawa,
Liouville Correlation Functions from Four-dimensional Gauge Theories,
{\em Lett. Math. Phys.} {\bf 91} (2010) 167--197, [arXiv:0906.3219].

\bibitem{AFKMY}
H.~Awata, H.~Fuji, H.~Kanno, M.~Manabe, and Y.~Yamada, 
Localization with a Surface Operator, Irregular Conformal Blocks and Open Topological String, 
{\em Adv.~Theor.~Math.~Phys.}\ {\bf 16} no.3 (2012), 
[arXiv:1008.0574]


\bibitem{BPZ}
A.~A.~Belavin, A.~M.~Polyakov and A.~B.~Zamolodchikov, 
Infinite conformal symmetry in two-dimensional quantum field theory, 
{\em Nucl. Phys.} {\bf B241}, (1984), 333-380

\bibitem{GIL}
O.~Gamayun, N.~Iorgov and O.~Lisovyy, 
Conformal field theory of Painlev\'e VI, [arXiv:1207.0787]

\bibitem{Har}
J,~Harnad, 
Quantum isomonodromic deformations and the Knizhnik-Zamolodchikov
equations. Symmetries and integrability of difference equations 
(Est\'erel, PQ, 1994),
155--161, CRM Proc. Lecture Notes, 9, {\em Amer.\ Math.\ Soc.,} Providence, RI, 1996

%\bibitem{FZ}
%A.~B.~Zamolodchikov and V.~A.~Fateev, 
%Operator algebra and correlation functions in the two-dimensional $\rm{SU}(2)\times \rm{SU}(2)$ chiral Wess-Zumino 
%model, 
%{\em Sov.~J.~Nucl.~Phys.~} {\bf 43} (4) (1986) 1031--1044 
%
%\bibitem{Giribet}
%G.~E.~Giribet, 
%Note on $\Z_2$ symmetries of the Knizhnik-Zamolodchikov equation, 
%{\em J.~Math.~Phys.~} {\bf 48} (2007)
%
%\bibitem{HF}
%Y.~Haraoka and G.~Filipuk, 
%Middle convolution and deformation for Fuchsian systems, 
%{\em J.~London.~Math.~Soc.~} (2) 76 (2007) 438--450

%\bibitem{HH}
%Y.~Haraoka and S.~Hamaguchi, 
%Topological theory for Selberg type integral associated with rigid Fuchsian systems, 
%{\em Math.~ Ann.~} to appear

\bibitem{JM}
M.~Jimbo and T.~Miwa, Monodromy preserving deformation of linear
ordinary diﬀerential equations with rational coeﬃcients. II, 
{\em Phys. 2D} 
(1981), 407--448.

\bibitem{JNS}
M.~Jimbo, H.~Nagoya and J.~Sun,
Remarks on the confuent KZ equation for $\slt$ and
quantum Painlev\'e equations,
{\em J. Phys. A: Math. Theor.} {\bf 41} (2008) 14pp

\bibitem{cubic}
K.~Kajiwara, T.~Masuda, M.~Noumi, Y.~Ohta and Y.~Yamada,
Cubic Pencils and Painlev\'e Hamiltonians,
{\em Funkcialaj Ekvacioj} {\bf 48} (2005) 147--160, 
[arXiv:nlin/0403009]

\bibitem{Kawakami}
H.~Kawakami, Generalized Okubo Systems and the Middle Convolution, 
{\em Int. Math. Res. Notices} {\bf 17} (2010)  3394--3421.


%\bibitem{Kohno}
%M.~Kohno, 
%Global Analysis in Linear Differential Equations, Kluwer, 1999
%

%\bibitem{K}
%G.~Kuroki, Quantum groups and quantization of Weyl group symmetries of Painlev\'e systems, {\em Adv.~ Stud.~ Pure Math.~} 
%to appear
%, arXiv:0808.2604 

\bibitem{K}
G.~Kuroki, Regularity of quantum $\tau$-functions generated
by quantum birational Weyl group actions, [arXiv:1206.3419]

%\bibitem{MR}
%K.~S.~Miller and B.~Ross, 
% An introduction to the fractional calculus and fractional differential equations, 
% Wiley-Interscience, 1993

\bibitem{NS}
N.~A.~Nekrasov and S.~L.~Shatashvili,
Quantization of Integrable Systems and Four Dimensional Gauge Theories,
[arXiv:0908.4052]

\bibitem{N QNY}
H.~Nagoya,
Quantum Painlev\'e Systems of Type $A_l^{(1)}$,
 {\em Int.\ J.\ Math.}\  {\bf 15} (2004), 1007--1031
 
% \bibitem{N higher}
% H.~Nagoya, 
% Quantum Painlev\'e Systems of Type $A_{n-1}^{(1)}$ with 
%higher degree Lax operators, {\em Int.\ J.\ Math. } {\bf 18} (2007), 
%no.~7,  839--868

\bibitem{NGR}
H. Nagoya, B. Grammaticos and A. Ramani, Quantum Painlev\'e  equations: from continuous to discrete and back, 
{\em Regular and Chaotic Dynamics} {\bf 13} (2008), no. 5, 417--423

\bibitem{N QPVI}
H.~Nagoya,
A quantization of the sixth Painlev\'e equation, 
Noncommutativity and singularities, {\em Adv. Stud. Pure Math.}\ {\bf 55} %{\em Math. Soc. Japan,} Tokyo 
(2009) 291--298  

\bibitem{Nagoya-Sun}
H.~Nagoya and J.~Sun,
Confluent primary fields in the conformal field theory, 
 {\em J. Phys. A: Math. Theor.}\ {\bf 43} 465203 (2010), [arXiv:1002.2598]

\bibitem{N HGS}
H.~Nagoya, 
Hypergeometric solutions to Schr\"odinger equations for the quantum Painlev\'e equations, 
{\em J. Math. Phys.}\ {\bf  52}  (2011) 16pp

\bibitem{N Weyl}
H. Nagoya, Realizations of affine Weyl group symmetries on  
the quantum Painlev\'e equations by fractional calculus,  
{\em Lett. Math. Phys.}\ {\bf 102}, no. 3, (2012), 297--321

\bibitem{Novikov}
D.~P.~Novikov, 
The $2\times 2$ matrix Schlesinger system and the Belavin-Polyakov-Zamolodchikov system, 
{\em Theor. Math. Phys.}\ {\bf 161} (2) (2009), 1485--1496 

%\bibitem{Novikov PVI}
%D.~P.~Novikov, Integral transformation of solutions for a Fuchsian-class equation 
%corresponding to the Okamoto transformation of the Painlev\'e VI equation, Theor. Math.
%Phys. Vol. 146 (2006), 295–303

\bibitem{O}
K.~Okamoto, 
Studies on the Painlev\'e equations, I:
{\em Ann.\ Math.\ Pura.\ Appl.}\ (4) {\bf 146} (1987),337--381;
II: {\em Jap.\ J.\ Math.}\ {\bf 13} (1987), no.~1, 47--76;
III: {\em Math.\ Ann.}\ {\bf 275} (1986), no.~2, 221--255;
IV: {\em Funkcial.\ Ekvac.}\ {\bf 30} (1987), no.~2-3, 305--332

\bibitem{OKSO}
Y.~Ohyama, H.~Kawamuko, H.~Sakai and K.Okamoto,
Studies on the Painlev\'e equations, V:
{\em J.~Math.~Sci.~Univ.~Tokyo} {\bf 13} (2006), 145--204

\bibitem{R}
N.~Reshetikhin,
The Knizhnik-Zamolodchikov system as a deformation of the isomonodromy probdfn,
{\em Lett. Math. Phys.}\ {\bf 26} (1992), 167--177

\bibitem{Suleimanov}
B. I. Suleimanov, 
The Hamiltonian structure of Painlev\'e equations and the method of isomonodromic deformations,
 {\em Differential Equations} {\bf 30} 726--732 (1994)



\bibitem{Takemura}
K.~Takemura, 
Integral representation of solutions to Fuchsian system and Heun's equation,
{\em J. Math. Anal. Appl.}\ {\bf 342} (2008) 52--69,  
[arXiv:0705.3358]


\bibitem{Teschner2}
J.~Teschner,
Quantization of the Hitchin moduli spaces, Liouville theory, and the geometric Langlands correspondence I, 
{\em Adv.~Theor.~Math.~Phys.}\ {\bf 15} no.~2 (2011) 471--564, 
[arXiv:1005.2846]

%\bibitem{Gaiotto}
%Davide Gaiotto,
%Asymptotically free $N=2$ theories and irregular conformal blocks,
%[arXiv:0908.0307]


%\bibitem{BMT}
%G.~Bonelli, K.~Maruyoshi and A.~Tanzini,
%Wild Quiver Gauge Theories,
%[arXiv:1112.1691]







\bibitem{ZZ}
A.~Zabrodin and A.~Zotov, 
Quantum Painlev\'e-Calogero Correspondence  
for Painlev\'e VI, 
{\em J. Math. Phys.}\ {\bf 53} (2012) 19pp, 
[arXiv:1107.5672]
 
%Comments: 7 pages


\end{thebibliography}
\end{document}